\documentclass[conference, compsoc]{IEEEtran}
\usepackage{array}
\usepackage{multirow}
\usepackage{amsthm}
\usepackage{amsmath}
\usepackage{amsfonts}
\usepackage{amssymb}
\usepackage{graphicx}
\usepackage{multirow}
\usepackage{setspace}
\usepackage{cite}
\usepackage{balance}
\usepackage{graphicx}
\usepackage{caption}
\usepackage{subcaption}
\usepackage{algorithmic}
\usepackage{algorithm}
\usepackage{comment}
\usepackage[labelformat=simple]{subcaption}

\newtheorem{thm}{Theorem}
\newtheorem{lem}{Lemma}

\newtheorem*{remark}{Remark}

\begin{document}

\title{\LARGE \textbf{Design and Analysis of Deadline and Budget Constrained Autoscaling (DBCA) Algorithm for 5G Mobile Networks}}
\author{Tuan Phung-Duc\IEEEauthorrefmark{2}, Yi Ren\IEEEauthorrefmark{1}, Jyh-Cheng Chen\IEEEauthorrefmark{1}, \textit{Fellow, IEEE}, and~Zheng-Wei Yu\IEEEauthorrefmark{1}\\
	\IEEEauthorrefmark{2}Faculty of Engineering, Information and Systems, University of Tsukuba, Ibaraki, Japan\\
		\IEEEauthorrefmark{1}Department of Computer Science, National Chiao Tung University, Hsinchu,  Taiwan, R.O.C.\\
	Emails:\IEEEauthorrefmark{2}tuan@sk.tsukuba.ac.jp, \IEEEauthorrefmark{1}\{renyi, jcc, zwyu12260\}@cs.nctu.edu.tw}

\maketitle

\begin{abstract}
In cloud computing paradigm, virtual resource autoscaling approaches have been intensively studied recent years. Those approaches dynamically scale in/out virtual resources to adjust  system performance for saving operation cost. However,  designing the autoscaling algorithm for desired performance with limited budget, while considering the existing capacity of legacy network equipment, is not a trivial task.  In this paper, we propose a Deadline and Budget Constrained Autoscaling (DBCA) algorithm for addressing the budget-performance tradeoff. We develop an analytical model to quantify the tradeoff and cross-validate the model  by extensive simulations. The results show that the DBCA can significantly improve system performance given the budget upper-bound. In addition, the model  provides a quick way to evaluate the budget-performance tradeoff and system design without wide deployment, saving on cost and time.
\end{abstract}

\IEEEpeerreviewmaketitle

\section{Introduction}

The emergence of Network Functions Virtualization (NFV) is changing the way of how mobile operators increase the capacities of their network infrastructures.  NFV offers fine-grained on-demand adjustment of network capabilities. Virtualized Network Functions (VNFs) can be scaled-out/in (turn on/off) to adjust  computing and network capabilities for saving energy and resources. A classic case is Animoto, an image-processing service provider, experienced a demand surging from 50 VM instances to 4,000 instances in three days, April 2008. After the peak, the demand fell sharply to an average level. Animoto only paid for 4,000 instances for the peak time~\cite{Animoto}.

Designing good \emph{autoscaling} strategies for budget constraints while meeting performance requirements is challenging.
In particular, operation cost is decreased by reducing the number of power-on VNF instances. On the other hand, resource under-provisioning may cause Service Level Agreements (SLAs) violations, leading to low Quality of user Experience (QoE). Therefore, the goal of desired autoscaling strategies is to meet the budget constraint while maintaining an acceptable level of performance. Then, a budget-performance tradeoff is formed:  The system performance is improved by adding more VNF instances while operation cost is reduced by the opposite way.

Designing autoscaling strategies for 5G mobile networks is different from that for traditional cloud computing scenarios. Specifically, in previous cloud autoscaling schemes (e.g.,~\cite{xiao2013dynamic,jokhio2013prediction,roy2011efficient,tirado2011predictive,  niu2012quality,calheiros2014workload, islam2012empirical,bankole2013cloud,  shen2011cloudscale,khan2012workload,gandhi2013exact,phung2014exact} ), only virtualized resources are considered. This is not suitable for typical cellular networks.  Given widely deployed existing legacy network equipment, the desired solution should consider the capacities of both legacy network equipment and VNFs. For example,  consider VNF only case that a VNF scaling-out from 1 VNF instance to 2 VNF instances increases 100\% capacity. Whereas, its capacity only grows less than 1\% if legacy network equipment (say 100 VNF instance capability) is counted. Current cloud autoscaling schemes usually ignore the non-constant issue.

In this paper, we investigate the budget-performance tradeoff in terms of deadline constraint, VM setup time, and the legacy equipment capacity.  We improve our recent work~\cite{YiRen2016globecom} by further considering deadline constraint for incoming requests, i.e., a request will be dropped if a pre-specified timer is expired. This is a more practical assumption compared with that in~\cite{YiRen2016globecom}, in which no deadline constraint is considered. To the best of our knowledge, this is the first work from this perspective. We then propose a Deadline and Budget Constrained Autoscaling (DBCA) algorithm for addressing the tradeoff. The DBCA considers available legacy equipment powered on all the time, while virtualized resources are divided into $k$ VNF instances. Then the DBCA scales out/in (turns on/off) VNF instances depending on job arrivals. Here, a central issue is how to choose a suitable $k$ for balancing the tradeoff. We then derive a detailed analytical model to answer this question. The analytical model quantifies the budget-performance tradeoff and cross-validates against extensive ns2 simulations. Furthermore, we propose a recursive approach to reduce the complexity of the computational procedure from $O(k^3K^3)$ to $O(kK)$ where $K$ the system capacity. Our model provides mobile operators with guidelines to design optimal VNF autoscaling strategies by their management policies in a systematical way, and enable wide applicability in various scenarios, and therefore, have important theoretical significance.

The rest of this paper is organized as follows. Section~\ref{sec:Related_Work} reviews the related work. Section~\ref{sec:Background} briefly introduces some background material on mobile networks and NFV architecture.
Section~\ref{sec:Proposed_Algorithm} presents the proposed optimal algorithm for VNF autoscaling applications. Section~\ref{sec:Analytical_Model} addresses the analytical models, followed by numerical results illustrated in Section~\ref{sec:Simulation_Results}.   Section~\ref{sec:Conclusions} offers conclusions.

\section{Related Work} \label{sec:Related_Work}
Recent years, autoscaling mechanisms have been intensively studied~\cite{xiao2013dynamic,jokhio2013prediction,roy2011efficient,tirado2011predictive,  niu2012quality,calheiros2014workload, islam2012empirical,bankole2013cloud, gandhi2013exact,phung2014exact, YiRen2016globecom, shen2011cloudscale,khan2012workload,mitrani2013managing,Mitrani20111222, mitrani2013trading, hu2015power,phung2015multiserver}. A straightforward and commonly used autoscaling approach is that autoscaling decisions are made based on resource utilization indicators (e.g., CPU, memory usage, etc). An example is  the default autoscaling approaches offered by Amazon EC2 and Microsoft Azure. A scale-out request is sent right way if the current CPU usage exceeds a predefined threshold. However, specifying the threshold value is not easy while considering VM setup time. Indeed, the setup lag time could be as long as 10 min or more to start an instance in Microsoft Azure; and the lag time could be various from time to time~\cite{hill2010early}. Thus it may happen that the instance is too late to serve the VNF so that one needs to leave more redundant while setting the threshold. To handle the setup time,  prediction/learning models are utilized to estimate the workload arrivals for autoscaling decision making, such as Exponential weighted Moving Average (EMA)~\cite{xiao2013dynamic,jokhio2013prediction}, Auto-Regressive Moving Average (ARMA)~\cite{roy2011efficient,tirado2011predictive}, Auto-Regressive Integrated Moving Average (ARIMA)~\cite{niu2012quality,calheiros2014workload}, machine learning~\cite{islam2012empirical, bankole2013cloud}, Markov model~\cite{shen2011cloudscale, khan2012workload}, recursive renewal reward~\cite{gandhi2013exact}, and matrix analytic method~\cite{phung2014exact}. However, \textit{the mechanisms~\cite{xiao2013dynamic,jokhio2013prediction,roy2011efficient,tirado2011predictive,  niu2012quality,calheiros2014workload, islam2012empirical,bankole2013cloud,  shen2011cloudscale,khan2012workload,gandhi2013exact,phung2014exact}  only consider virtualized resource itself (cloud resource) while overlooking legacy (fixed) resources}, which are not suitable for typical cellular networks.

Perhaps the closest models to ours were studied in~\cite{YiRen2016globecom, mitrani2013managing,Mitrani20111222, mitrani2013trading, hu2015power,phung2015multiserver} that both the capacities of fixed legacy network equipment and dynamic autoscaling cloud servers are considered. The authors in~\cite{mitrani2013managing,Mitrani20111222} consider setup time without defections~\cite{mitrani2013managing} and with defections~\cite{Mitrani20111222}. Our recent work~\cite{hu2015power} relaxes the assumption in~\cite{mitrani2013managing,Mitrani20111222} that after a setup time, all the cloud servers in the block are active concurrently. We further consider a more realistic model that each server has an independent setup time. However, in~\cite{mitrani2013managing,Mitrani20111222, hu2015power}, all the cloud servers were assumed as a whole block, which is not practical where each cloud server should be allowed to scale-out/in dynamically. Considering all cloud servers as a whole block was relaxed to sub-blocks in~\cite{mitrani2013trading,phung2015multiserver}. However, either setup time is ignored~\cite{mitrani2013trading}, or fixed legacy network capacity is not considered~\cite{phung2015multiserver}. Our recent work~\cite{YiRen2016globecom} fixes the research gap, whereas job deadline constraint is not considered.

\section{Background} \label{sec:Background}
\begin{figure}
	\centering
	\includegraphics[width=9cm]{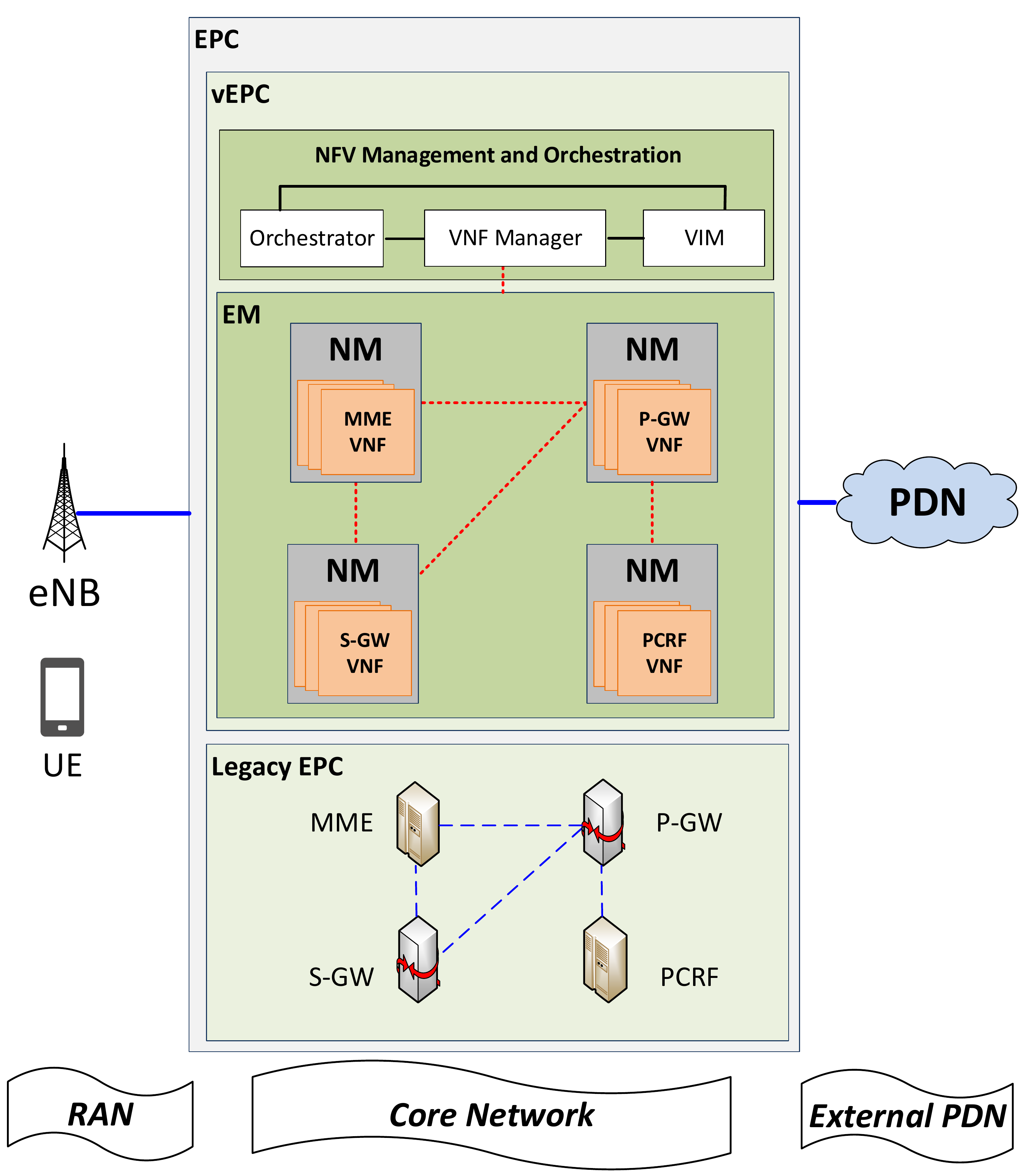}
	\caption{A simplified example of NFV enabled LTE architecture.}
	\label{fig:vEPC_arch}
	\vspace{-5mm}
\end{figure}
Mobile Core Network (CN) is one of the most important parts in mobile networks. The main target of NFV is to virtualize  the functions in the CN. The most recent CN is the Evolved Packet Core (EPC) introduced in Long Term Evolution (LTE). Here, we use an example to explain EPC and virtualized EPC (vEPC) when NFV is deployed. Fig.~\ref{fig:vEPC_arch} shows a simplified example of NFV enabled LTE architecture consisted of Radio Access Network (RAN), EPC, and external Packet Data Network (PDN). In particular,  the EPC is composed of legacy EPC and vEPC. In the following, we brief introduce them respectively.

\subsection{Legacy EPC and vEPC}
EPC is the CN of the LTE system. Here, we only show basic network functions, such as Serving Gateway (S-GW), PDN Gateway (P-GW), Mobility Management Entity (MME), and Policy and Charging Rules Function (PCRF) in the EPC.


To virtualize the above network functions, 3GPP introduces NFV management functions and solutions for  vEPC based on ETSI NFV specification~\cite{V13.1.02015}, as shown in Fig.~\ref{fig:vEPC_arch}. The network functions (e.g., MME, PCRF) are denoted as Network Elements (NE), which are virtualized as VNF instances. Network Manager (NM) provides end-user functions for network management of NEs. Element Manager (EM) is responsible for the management of a set of NMs. NFV management and orchestration controls VNF instance scaling procedure, which are detailed as follows.
\begin{itemize}
	\item VNF scale-in/out: VNF scale-out adds additional VMs to support a VNF instance, adding more virtualized hardware resources (i.e., compute, network, and storage capability) into the VNF instance. In contrast, VNF scale-in removes existing VMs from a VNF instance.
	\item VNF scale-up/down: VNF scale-up allocates more hardware resources into a VM for supporting a VNF instance (e.g., replace a One-core with Dual-core CPU). Whereas, VNF scale-down releases hardware resources from a VNF instance.
\end{itemize}

\section{Proposed Deadline and Budget Constrained Autoscaling Algorithm} \label{sec:Proposed_Algorithm}

\subsection{System Model and DBCA: Deadline and Budget Constraint Autoscaling} \label{ssec:DBCA}
In general, we consider that a 5G EPC consists of legacy network entities (e.g., MME, PCRF) and VNFs~\cite{hawilo2014nfv,chih2014toward}. For a network entity, its capacity is supported by both legacy network equipment and VNF instances. Fig.~\ref{fig:Queueing_model_special_case} illustrates a simplified example of a network entity queueing model considering the capacities of both VNF instances and legacy network equipment. Specifically, the capacity of its legacy network equipment is assumed to be $n_0$ VNF instance capacities while $k$ denotes the number of VNF instances for supporting the network entity. That is, the total capacity of the network entity is $k+n_0=N$.

From the network entity's point of view, we assume that user requests arrive according to a Poisson process with rate $\lambda$. The capacity of a VNF instance is assumed to accept one job at a time and the service time is assumed to follow the exponential distribution with mean $1/\mu$. When a user request arrives, the job first enters a limited First-Come-First-Served (FCFS) queue waiting for processing. Each job has deadline constraint, which is a random variable following the exponential distribution with mean $1/\theta$. In other words, the job will quit the queue if its waiting time exceeds its deadline. Without loss of generality, the legacy network equipment is always on while VNF instances will be powered on (or off) according to the number of waiting jobs in the queue. Moreover, a VNF instance needs a setup time to be available to serve user requests, which is assumed to be an exponentially distributed random variable with mean value $1/\alpha$.

\begin{figure}[t]
	\centering
	\includegraphics[width=8.7cm]{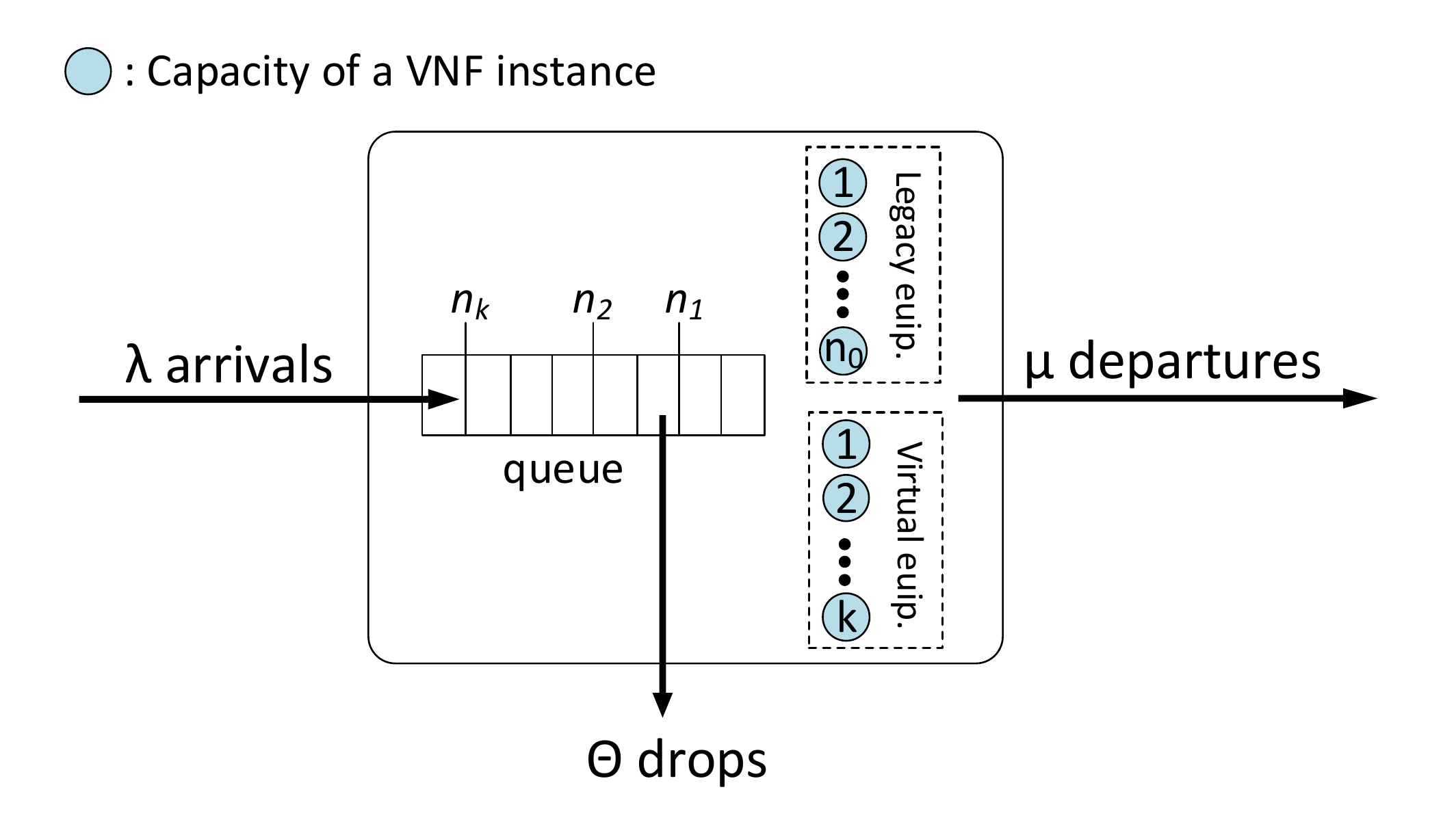}
	\caption{A service center with reserve blocks.}
	\label{fig:Queueing_model_special_case}
	\vspace{-3mm}
\end{figure}

 DBCA utilizes two thresholds, '\texttt{Up}' and '\texttt{Down}', or $U_i$ and $D_i$ to control the VNF instances $i=1,2,\cdots, k$.   Further, let $n_1=n_0+1$ and $n_i=n_{i-1}+1$ ($i=1,2,\cdots, k$), i.e., $n_k=N$. In other words, DBCA sends orders to NFV management and orchestration to turn on/off VNF instances to adjust network capacities.
 \begin{itemize}
 	\item \textit{$U_i$, power up the $i$-th VNF instances:}  If the $i$-th VNF instance is turned off and the number of requests in the system increases from $U_i-1$ to $U_i$, then the VNF instance is powered up after a setup time to support the system.  During the setup time, a VNF instance cannot serve user requests, but consumes power (or money for renting cloud services). Here, we specify $U_i=n_i$. It is equivalent to that when the number of requests increases from $n_{i-1}$ to $n_{i}$, the $i$-th VNF instance is powered up.
 	\item \textit{$D_i$, power down the $i$-th VNF instances:}  If the $i$-th VNF instance is operative, and the number of requests in the system drops from $D_i+1$ to $D_i$, then the VNF instance is powered down instantaneously. In this paper, we choose $D_i = n_{i-1}$. It is equivalent to that when the number of requests drops from $n_i$ to $n_{i-1}$, we turn off the $i$-th VNF instance.
 \end{itemize}

\subsection{Performance Metrics} \label{ssec:PerformMetrics}
The system performance is evaluated by four metrics: the average response time in the queue $W_q$, the average number of running VNF instance $S$, user request blocking probability $P_b$, and user request dropping probability $P_d$.  We define them as follows.
\begin{itemize}
  \item \textit{The average response time in queue $W_q$} is defined as a job request's waiting time in queue. In other words, it reveals how long time a job request can be served.
  \item \textit{The average number of running VNF instances $S$} addresses the operation cost of virtual equipment.
  \item \textit{Dropping probability $P_d$}  is the probability that a request's waiting time in queue exceeds its deadline constraint.
  \item \textit{Blocking probability $P_b$}  is the probability that a request is denied due to system busy.
\end{itemize}

The closed-form solutions of $W_q$, $S$, $P_b$, and $P_d$ are given as (\ref{eq:W_q}), (\ref{eq:S}), (\ref{eq:Pb}), and (\ref{eq:Pd}) in Section~\ref{sec:Analytical_Model}. Thus, the system performance $P$ has the form
\begin{align}\label{eq:SysPerform}
	P=w_1W_q+w_2S+w_3P_b+w_4P_d,
\end{align}
where coefficients $w_1$, $w_2$, $w_3$, and $w_4$ denote the weight factors for $W_q$, $S$, $P_b$, and $P_d$,  respectively. Increasing $w_1$ (or $w_2$, $w_3$, $w_4$) emphasizes more on $W_q$ (or $S$, $P_b$, $P_d$). Here, we do not specify either $w_1$ or $w_2$ ($w_3$, $w_4$) due to the fact that such a value should be determined by a mobile operator and must take management policies into consideration.

\section{Analytical Model} \label{sec:Analytical_Model}

\begin{table}[t]
	\small
	\caption{List of Notations}
	\centering  \label{tab:ParameterSetting}
	\begin{tabular}{|c|p{6.7cm}|}
		\hline
		\textbf{Notation} & \textbf{Explanation }                                         \\ \hline\hline
		$N$           & The total capacities of a network entity                            \\ \hline
		$K$            & The number of maximum jobs can be accommodated in the system \\ \hline
		$k$                & The number of VNF instances   \\ \hline
		$P$            & System performance \\ \hline
		$W$            & Average response time  \\ \hline
		$W_q$            & Average response time in queue  \\ \hline
		$S$                & Average VM cost  \\ \hline
		$P_b$      & Blocking probability \\ \hline
		$P_d$ & Dropping probability \\ \hline
		$w_1$       & Weight factor for $W_q$   \\ \hline
		$w_2$        & Weight factor for $S$     \\ \hline
		$w_3$& Weight factor for $P_b$ \\ \hline
		$w_4$ & Weight factor for $P_d$  \\ \hline
		$n_{0}$      & The capacities of legacy network equipment                      \\ \hline
		$U_{i}$       & The \texttt{Up} threshold to control the VNF instances\\ \hline
		$D_{i}$       & The \texttt{Down} threshold to control the VNF instances \\ \hline
		$m_{i}$      & The $i$-th reserve sub-block ($i = 1, 2, \cdots k$).                  \\ \hline
		$\lambda$ & Job arrival rate                                \\ \hline
		$\mu$       & Service rate for each server                    \\ \hline
		$\alpha$     & Setup rate for each virtual server              \\ \hline
		$\theta$     & Abandonment rate of each job              \\ \hline

	\end{tabular}
\end{table}

In this section, we propose the analytical model for DBCA. The goal of the analytical model is to cross-validate the accuracy of the simulation experiments and to analyze both the operation cost and the system performance for DBCA.  Given the analytical model, one can quickly obtain the operation cost and system performance for DBCA, without real deployment, saving on cost and time.

We model the system as a queueing model with $N$ servers and a capacity of $K$, i.e., the maximum of $K$ jobs can be accommodated in the system. Job arrivals follow  the Poisson distribution with rate $\lambda$. A VNF instance (server) accepts one job at a time, and its service time follows the exponential distribution with mean $1/\mu$. There is a limited FCFS queue for those jobs that have to wait for processing.

In this system, a server is turned off immediately if it has no job to do. Upon arrival of a job, an OFF server is turned on if any and the job is placed in the buffer. However, a server needs some setup time to be active so as to serve waiting jobs. We assume that the setup time follows the exponential distribution with mean $1/\alpha$. Let $j$ denotes the number of customers in the system and $i$ denotes the number of active servers. The number of reserves (server) in setup process is $\min(j-n_i, N-n_i)$. Here, $n_i=n_{i-1}+m_i$, where $m_i=1$ for all $i$ (block size is one). Therefore, in this model a server in reserve blocks is in either BUSY or OFF or SETUP.  We assume that waiting jobs are served according to an FCFS manner. We call this model an $M/M/N/K$ Setup queue.


Here, we present a recursive scheme to calculate the joint stationary distribution.
Let $C(t)$ and $L(t)$ denote the number of active servers and the number of customers in the system, respectively. It is easy to see that $\{ X(t) = (C(t),L(t)); t \geq 0\}$ forms a Markov chain on the state space:
\begin{align*}
\mathcal{S} =  &   \{(i,j); 1 \leq i \leq k, j  = n_i,n_i+1,\dots,K-1,K\}\\
& \cup \{(0,j); j = 0,1,\dots,K-1,K\}.
\end{align*}

\begin{figure}[t]
	\begin{center}
		\includegraphics[width=8.6cm]{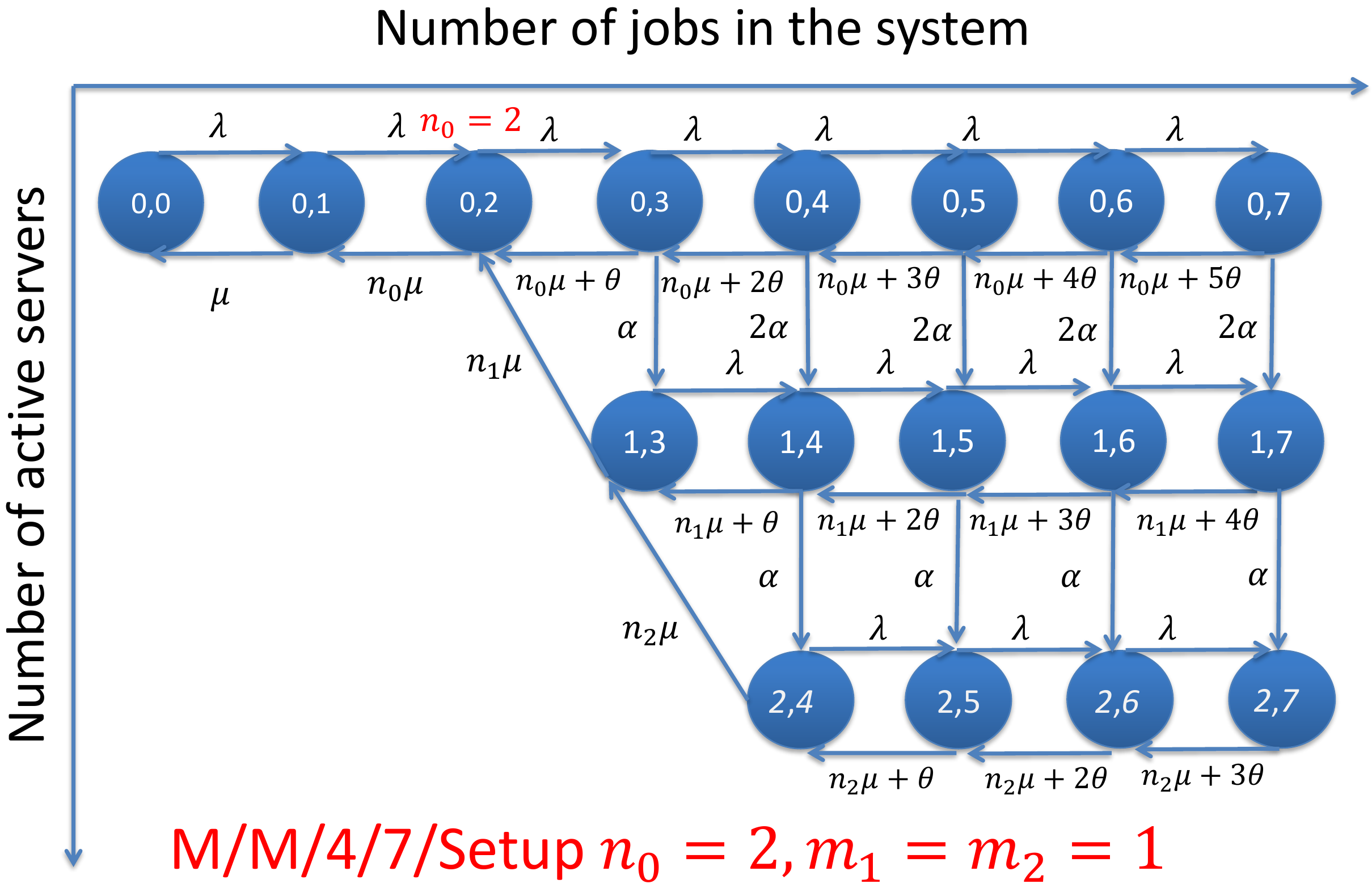}
		\caption{Transition among states ($N=4, n_0, = 2, m_1 = m_2 = 1$,  and $K=7$).}
		\label{fig:Markov_Chain_4_7}
	\end{center}
	\vspace{-11mm}
\end{figure}

Fig.~\ref{fig:Markov_Chain_4_7} shows the transition among states for the case  where $N=4, n_0 = 2, m_1 = m_2 = 1$,  and $K=7$.  Let $\pi_{i,j} = \lim_{t \to \infty} {\rm P} (C(t) = i, L(t) = j)$ ($(i,j) \in \mathcal{S}$) denote the joint stationary distribution of $\{X(t) \}$.
Here, we derive a recursion for calculating the joint stationary distribution $\pi_{i,j}$ ($(i,j) \in \mathcal{S}$). The balance equations for states with $i=0$ read as follows.
\begin{align*}
\lambda \pi_{0,j-1} \textit{} &=  j \mu \pi_{0,j},\\
&\mathrm{for}\;\; j = 0,1,\dots,n_0,  \\
\lambda \pi_{0,j-1}  + n_0 \mu \pi_{0,j+1} &=  (\lambda + \min (j-n_0,N-n_0) \alpha + n_0 \mu    ) \pi_{0,j} ,\\
&\mathrm{for}\;\; j = n_0,n_0 + 1, \dots, K-1, \\
\lambda \pi_{0,K-1}& =   (n_0 \mu + (N-n_0) \alpha) \pi_{0,K},
\end{align*}
leading to
\begin{align*}
\pi_{0,j} = b^{(0)}_j \pi_{0,j-1}, \qquad j=1,2,\dots, K.
\end{align*}
The sequence, $\{ b^{(0)}_j; j = 1,2,\dots,K \}$ is given as follows.
\[
b^{(0)}_j = \frac{\lambda}{j\mu}, \qquad j = 1,2,\dots,n_0,
\]
and
\begin{align*}
b^{(0)}_j = \frac{\lambda}{A1_j},\qquad j = K-1,K-2,\dots,n_0 + 1,
\end{align*}
where $A1_j=\lambda + n_0 \mu + \min (j-n_0,N-n_0) \alpha + (j-n_0) \theta - (n_0 \mu + (j+1-n_0) \theta) b^{(0)}_{j+1}$ and
\[
b^{(0)}_K = \frac{\lambda}{ n_0 \mu + (N-n_0) \alpha + (K-n_0) \theta }.
\]
Furthermore, it should be noted that $\pi_{1,1}$ is calculated using the local balance equation in and out the set $\{(0,j); j = 0,1,\dots,K\}$ as follows.
\[
n_1 \mu \pi_{1,1} = \sum_{j=n_1}^K \min (j,N-n_0) \alpha \pi_{0,j}.
\]

\begin{remark}
	We have expressed $\pi_{0,j}$ ($j = 1,2,\dots,K$) and $\pi_{1,1}$ in terms of $\pi_{0,0}$.
\end{remark}

%

%
%

Next, we consider the case $i=1$.

\begin{lem}\label{lemma:pi1j}
	We have
	\[
	\pi_{1,j}  = a^{(1)}_j + b^{(1)}_j \pi_{1,j-1}, \qquad j = 2,3, \dots, K-1,K,
	\]
	where
	\begin{align}
		\label{a1_j:eq}
\nonumber		a^{(1)}_j  = & \frac{1}{A2} \Big((n_1 \mu + (j+1-n_1) \theta) a^{(1)}_{j+1}\\
 &\qquad+ \min(j-n_0,N-n_0) \alpha \pi_{0,j} \Big)  , \\
		\label{b1_j:eq}
		b^{(1)}_j  =&  \frac{ \lambda   }{A2},
	\end{align}
	for
\begin{align*}
j &= K-1, K-2, \dots,2,\\
A2&= n_1 \mu + \lambda + \min(j-n_1,N-n_1) \alpha\\
& \quad+  (j-n_1) \theta - (n_1 \mu + (j+1-n_1)\theta) b^{(1)}_{j+1},\\
a^{(1)}_K &= \frac{(N-n_0) \alpha \pi_{0,K}}{n_1 \mu + (N-n_1) \alpha + (K-n1) \theta },\\
\qquad b^{(1)}_K &= \frac{\lambda}{n_1 \mu + (N-n_1) \alpha + (K-n_1)\theta}.
\end{align*}

\end{lem}

\begin{proof}
	We prove using mathematical induction. Balance equations are given as follows.
	\begin{align}
		\label{pi1j:eq}
		&(\lambda + n_1 \mu + \min ( j-n_1, N-n_1 ) \alpha) \pi_{1,j}  \nonumber \\
		&=  \lambda \pi_{1,j-1} + n_1 \mu \pi_{1,j+1} + \min(j-n_0,N-n_0) \alpha \pi_{0,j}, \\
		&\mathrm{for}\;\; 2 \leq j \leq K-1, \nonumber \\
		\label{pi1K:eq}
\nonumber	&(\mu + \min (K-n_1,N-n_1) \alpha\\
&\quad + (N-n_1)\theta ) \pi_{1,K}  =  \lambda \pi_{1,K-1} + (N-n_0) \alpha \pi_{0,K}.
	\end{align}
	It follows from (\ref{pi1K:eq}) that
	\[
	\pi_{1,K} = a^{(1)}_K + b^{(1)}_K \pi_{1,K-1} ,
	\]
	leading to the fact that Lemma~\ref{lemma:pi1j} is true for $j=K$. Assuming that Lemma~\ref{lemma:pi1j} is true for $j+1$, i.e., $\pi_{1,j+1}  = a^{(1)}_{j+1} + b^{(1)}_{j+1} \pi_{1,j}$. It then follows from (\ref{pi1j:eq}) that Lemma~\ref{lemma:pi1j} is also true for $j$, i.e., $\pi_{1,j}  = a^{(1)}_j + b^{(1)}_j \pi_{1,j-1}$.
\end{proof}
\begin{thm}\label{theorem1:thm}
	We have the following bound.
	\[
	a^{(1)}_j \geq 0, \qquad 0 \leq b^{(1)}_j  \leq \frac{\lambda}{n_1 \mu  + \min(j-n_1,N-n_1) \alpha},
	\]
	for $j = 2,3,\dots,K-1,K$.
\end{thm}
\begin{proof}
	We use mathematical induction. It is easy to see that the theorem is true for $j=K$. Assuming that the theorem is true for $j+1$, i.e.,
	\begin{align*}
		0&\leq a^{(1)}_{j+1} , \\
		\qquad 0 &\leq b^{(1)}_{j+1}  \leq \frac{\lambda}{n_1 \mu  + \min(j-n_1,N-n_1) \alpha + (j-n_1) \theta },
	\end{align*}where $j = 1,2,\dots,K-1.$

	Thus, we have $\mu b^{(1)}_{j+1} < \lambda$. From this inequality, (\ref{a1_j:eq}) and (\ref{b1_j:eq}), we obtain
	\[
	b^{(1)}_j  \leq \frac{\lambda}{n_1 \mu  + \min(j-n_1,N-n_1) \alpha + (n-n_1) \theta},
	\]
	and $a^{(1)}_j \geq 0$.
\end{proof}
It should be noted that $\pi_{2,2}$ can be calculated using the local balance between the flows in and out the set of states $\{ (i,j); i = 0,1, j = i, i+1,\dots,K\}$ as follows.
\[
n_2\mu \pi_{2,n_2} = \sum_{j=n_2}^K  \min(j-n_1,N-n_1) \alpha \pi_{1,j}.
\]
\begin{remark}
	We have expressed $\pi_{1,j}$ ($j = 1,2\dots,K$) and $\pi_{2,2}$ in terms of $\pi_{0,0}$.
\end{remark}

We consider the general case where $2 \leq i \leq k-1$.
Similar to the case $i=1$, we can prove the following result by mathematical induction.
\begin{lem}\label{lemma:pi_ij}
	We have
	\[
	\pi_{i,j} = a^{(i)}_j + b^{(i)}_j \pi_{i,j-1}, \qquad j = i + 1, i +2, \dots, K-1,K,
	\]
	where
	\begin{align}\label{backward:i}
		\nonumber a^{(i)}_j  & =  \frac{1}{A3}\Big( (n_i\mu + (j+1-n_i) \theta ) a^{(i)}_{j+1}\\
&\qquad \qquad +  \min(N-n_{i-1},j-n_{i-1}) \alpha \pi_{i-1,j}  \Big), \\
		\label{backward:ib}
		b^{(i)}_j  & =  \frac{\lambda}{A3},
	\end{align}
	and
\begin{align*}
 A3&= \lambda + \min (N-n_i,j-n_i) \alpha + n_i \mu\\
 & \quad+ (j-n_i) \theta -( n_i \mu + (j+1-n_i)\theta ) b^{(i)}_{j+1} \\
 a^{(i)}_K &= \frac{(N-n_{i-1}) \alpha \pi_{i-1,K}}{(N-n_i) \alpha + n_i \mu + (K-n_i) \theta },\\
 b^{(i)}_K &= \frac{\lambda }{(N-n_i) \alpha + n_i \mu  + (K-n_i) \theta }.
\end{align*}

\end{lem}
\begin{proof}
	The balance equation for state $(i,K)$ is given as follows.
	\[
	((N-n_i) \alpha + n_i \mu (K-n_i) \theta ) \pi_{i,K} = \lambda \pi_{i,K-1} + (c-n_{i-1}) \alpha \pi_{i-1,K},
	\]
	leading to the fact that Lemma~\ref{lemma:pi_ij} is true for $j=K$.
	Assuming that
	\[
	\pi_{i,j+1} = a^{(i)}_{j+1} + b^{(i)}_{j+1} \pi_{i,j}, \qquad j = i + 1, i +2, \dots, K-1.
	\]
	It then follows from
	\begin{align*}
		& (\lambda + \min (N-n_i,j-n_i) \alpha + n_i \mu  + (j-n_i) \theta ) \pi_{i,j} \\
		& =  \lambda \pi_{i,j-1} + (n_i \mu  + (j+1-n_i) \theta ) \pi_{i,j+1} \\
&\qquad+ \min(N-n_{i-1},j-n_{i-1}) \alpha \pi_{i-1,j}, \\
		& j = K-1,K-2,\dots,i+1,
	\end{align*}
	that
	\[
	\pi_{i,j} = a^{(i)}_{j} + b^{(i)}_{j} \pi_{i,j-1}.
	\]
\end{proof}

\begin{thm}\label{theorem:i}
	We have the following bound.
\begin{align*}
	&a^{(i)}_j >0,\\
	&0 < b^{(i)}_j  < \frac{\lambda}{n_i \mu  + \min(j-n_i,N-n_i) \alpha + (j-n_i) \theta},
\end{align*}
for $j = n_i+1,n_i+2,\dots,K, i = 1,2,\dots,k-1$.
\end{thm}
\begin{proof}
	We also prove using mathematical induction. It is clear that Theorem~\ref{theorem:i} is true for $j=K$. Assuming that Theorem~\ref{theorem:i} is true for $j+1$, i.e.,
\begin{align*}
	&a^{(i)}_{j+1} >0,\\
 & 0 < b^{(i)}_{j+1}  < \frac{\lambda}{n_i\mu  + \min(j+1-n_i,N-n_i) \alpha},
\end{align*}

	for $j = i+1,i+2,\dots,K-1, i = 1,2,\dots,c-1$.
	It follows from the second inequality that $i\mu b^{(i)}_{j+1} < \lambda$. This together with formulae (\ref{backward:i}) and (\ref{backward:ib}) yield the desired result.
\end{proof}

It should be noted that $\pi_{i+1,i+1}$ is calculated using the following local balance equation in and out the set of states:
\[ \{(k,j); k=0,1,\dots,i; j = k,k+1,\dots,K \}\] as follows.
\[
n_{i+1} \mu \pi_{i+1,i+1} = \sum_{j=n_i+1}^K \min(j-n_i,N-n_i) \alpha \pi_{i,j}.
\]
\begin{remark}
	We have expressed $\pi_{i,j}$ ($i=0,1,\dots,c-1, j = i,i+1,\dots,K$) and $\pi_{i+1,i+1}$ in terms of $\pi_{0,0}$.
\end{remark}

Finally, we consider the case $i = k$. Balance equation for state $(k, K)$ yields,
\begin{lem}
	We have
	\[
	\pi_{k,j} = a^{(k)}_j + b^{(k)}_j \pi_{k,j-1}, \qquad j = n_k + 1, n_k +2, \dots, K,
	\]
	where
	\begin{align}\label{backward:c}
		&a^{(k)}_j  =  \frac{ (n_k \mu + (j+1-n_k)\theta) a^{(k)}_{j+1} + (N-n_{k-1})\alpha \pi_{k-1,j}}{\lambda + n_k\mu + (j-n_k) \theta - (n_k \mu + (j + 1-n_k) \theta) b^{(k)}_{j+1}},\\
		& \qquad j = K-1, K-2, \dots, n_k+1,\nonumber\\
		\label{backward:cb}
		&b^{(k)}_j  =  \frac{\lambda}{\lambda + n_k\mu + (j-n_k) \theta - (n_k \mu + (j+1-n_k) \theta) b^{(k)}_{j+1}}, \\
		&\qquad j = K-1, K-2, \dots, n_k+1,\nonumber
	\end{align}
	and
	\[
	a^{(k)}_K = \frac{\alpha \pi_{k-1,K}}{n_k \mu + (j-n_k) \theta }, \qquad b^{(k)}_K = \frac{\lambda}{n_k \mu + (j-n_k) \theta }.
	\]
\end{lem}
\begin{proof}
	The global balance equation at state $(k,K)$ is given by
	\[
	(n_k \mu + (j-n_k) \theta) \pi_{k,K} = (N-n_{k-1}) \alpha \pi_{k-1,K} + \lambda \pi_{k,K-1},
	\]
	leading to
	\[
	\pi_{k,K} = a^{(k)}_K + b^{(k)}_K \pi_{k,K-1}.
	\]
	Assuming that $\pi_{k,j+1} = a^{(k)}_{j+1} + b^{(k)}_{j+1} \pi_{k,j}$, it follows from the global balance equation at state $(k,j)$,
	\begin{align*}
		(\lambda + &n_k \mu + (j-n_k)\theta ) \pi_{k,j} = \lambda \pi_{k,j-1} \\
&+ (n_k \mu + (j+1-n_k)\theta)  \pi_{n_k,j+1} + (N-n_{k-1}) \alpha \pi_{k-1,j},\\
		&\qquad  j = n_k+1,n_k+2,\dots,K-1,
	\end{align*}that $\pi_{k,j} = a^{(k)}_j + b^{(k)}_j \pi_{k,j-1}$ for $j=n_k+1,n_k+2,\dots,K$.
\end{proof}
\begin{thm}\label{ab_cj:thm}
	We have the following bound.
\begin{align*}
	a^{(k)}_j >0,\qquad &0 < b^{(k)}_j  < \frac{\lambda}{n_k \mu + (j-n_k) \theta}, \\
 &j = n_k+1,n_k+2,\dots,K-1.
\end{align*}

\end{thm}
\begin{proof}
	We also prove using mathematical induction. It is clear that Theorem~\ref{ab_cj:thm} is true for $j=K$. Assuming that Theorem~\ref{ab_cj:thm} is true for $j+1$, i.e.,
	\begin{align*}
		a^{(k)}_{j+1} >0, \qquad &0 < b^{(k)}_{j+1}  < \frac{\lambda}{n_k \mu + (j-n_k) \theta}, \\
		&j = n_k+1,n_k+2,\dots,K-1.
	\end{align*}
	
	It follows from the second inequality that $n_k \mu b^{(k)}_{j+1} < \lambda$. This together with formulae (\ref{backward:c}) and (\ref{backward:cb}) yield the desired result.
\end{proof}

We have expressed all the probability $\pi_{i,j}$ ($(i,j) \in \mathcal{S}$) in terms of $\pi_{0,0}$ which is uniquely determined by the normalizing condition.
\[
\sum_{(i,j) \in \mathcal{S}} \pi_{i,j} = 1.
\]

Let ${\rm E}[L]$ denote the mean number of jobs in the system. We have
\[
{\rm E}[L] =\sum_{(i,j) \in \mathcal{S}} \pi_{i,j} j = \sum_{i=0}^{n_0-1} \pi_{0,j} j + \sum_{i=0}^k \sum_{j=n_i}^K \pi_{i,j} j.
\]

Let $P_b$ denote the blocking probability. We have
\begin{align}\label{eq:Pb}
P_b = \sum_{i=0}^k \pi_{i,K}.
\end{align}

It follows from Little's law that
\begin{align}
W = \frac{{\rm E}[L]}{\lambda (1-P_b)}=\frac{\sum_{i=0}^{n_0-1} \pi_{0,j} j + \sum_{i=0}^k \sum_{j=n_i}^K \pi_{i,j} j}{\lambda (1-\sum_{i=0}^k \pi_{i,K})}.
\end{align}
We obtain
\begin{align}
	W_q = W-\frac{1}{\mu}. \label{eq:W_q}
\end{align}

The mean number of VNF instances is given by
\begin{align}
	S = \sum_{(i,j) \in \mathcal{S}} \pi_{i,j} (n_i-n_0) + \sum_{i=0}^k \sum_{j = n_i}^K \pi_{i,j} \min(j-n_i,N-n_i), \label{eq:S}
\end{align}
where the first term is the number of VNF instances that are already active while the second term is the mean number of VNF instances in setup mode.

Let ${\rm E}[Q]$ denote the mean number of waiting jobs in the system. We have
\[
	{\rm E}[Q] = \sum_{i=0}^k \sum_{j = n_i}^K \pi_{i,j} (j-i).
\]
Let $P_d$ denote the reneging probability that a waiting job abandons from the system. We have
\begin{align}\label{eq:Pd}
	P_d & = \frac{{\rm E}[Q] \theta}{\lambda (1-P_b)}=\frac{\sum_{i=0}^k \sum_{j = n_i}^K \pi_{i,j} (j-i)\theta}{\lambda (1-P_b)},
\end{align}
where the numerator and the denominator are the abandonment rate and the arrival rate of accepted jobs, respectively.

Again, based on the above derived performance metrics $W_q$, $S$, $P_b$, and $P_d$, mobile operators can easily design network optimization strategies according to~(\ref{eq:SysPerform}).
\begin{remark}
Theorems 1, 2, and 3 allow us to calculate the joint stationary distribution by a numerically stable algorithm because we deal with only positive numbers.
\end{remark}

\section{Simulation and Numerical Results} \label{sec:Simulation_Results}
This section provides both simulation and numerical results for the analytical model addressed in Section~\ref{sec:Analytical_Model}. The analytical model is cross-validated by extensive simulations by using ns2, version 2.35~\cite{ns2} with real measurement results for parameter configuration\footnote{Due to simulation time limitation, $\lambda$ and $\mu$ are scaled down accordingly with the same ratio $\lambda/\mu$.}:  $\lambda$ by Facebook data center traffic~\cite{lambda}, $\mu$ by the base service rate of a Amazon EC2 VM~\cite{gilani2015application}, and $\alpha$ by the average VM startup time~\cite{alpha}.  If not further specified, the following parameters are set as the default values for performance comparison: $n_0 = 110$, $\mu = 1$, $\alpha = 0.005$, $K = 250$, $\lambda = 50 \thicksim 250$ (see Table 1 for details). The results are based on exponential distribution for job request inter arrival time and VNF instance service time with mean $1/\lambda$ and $1/\mu$. The simulation time is 300,000 seconds. And $15 \thicksim 75$ millions job requests were generated during the extensive simulations.

\begin{figure*}
\vspace{-3mm}
	\centering
	\begin{subfigure}[b]{0.45\textwidth}
		\includegraphics[width=8cm]{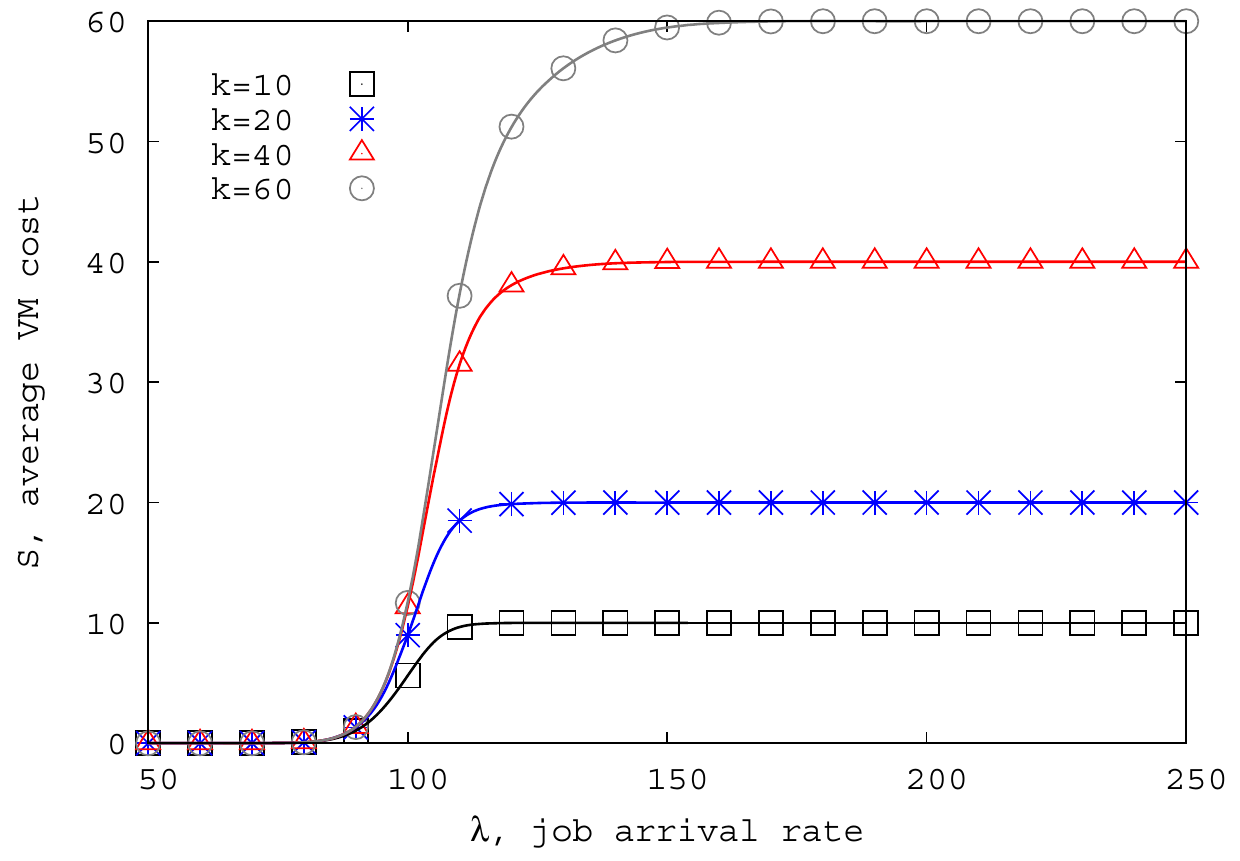}
		\caption{Impacts  on $S$.}
		\label{sfig:S_k}
	\end{subfigure}
		\begin{subfigure}[b]{0.45\textwidth}
			\includegraphics[width=8cm]{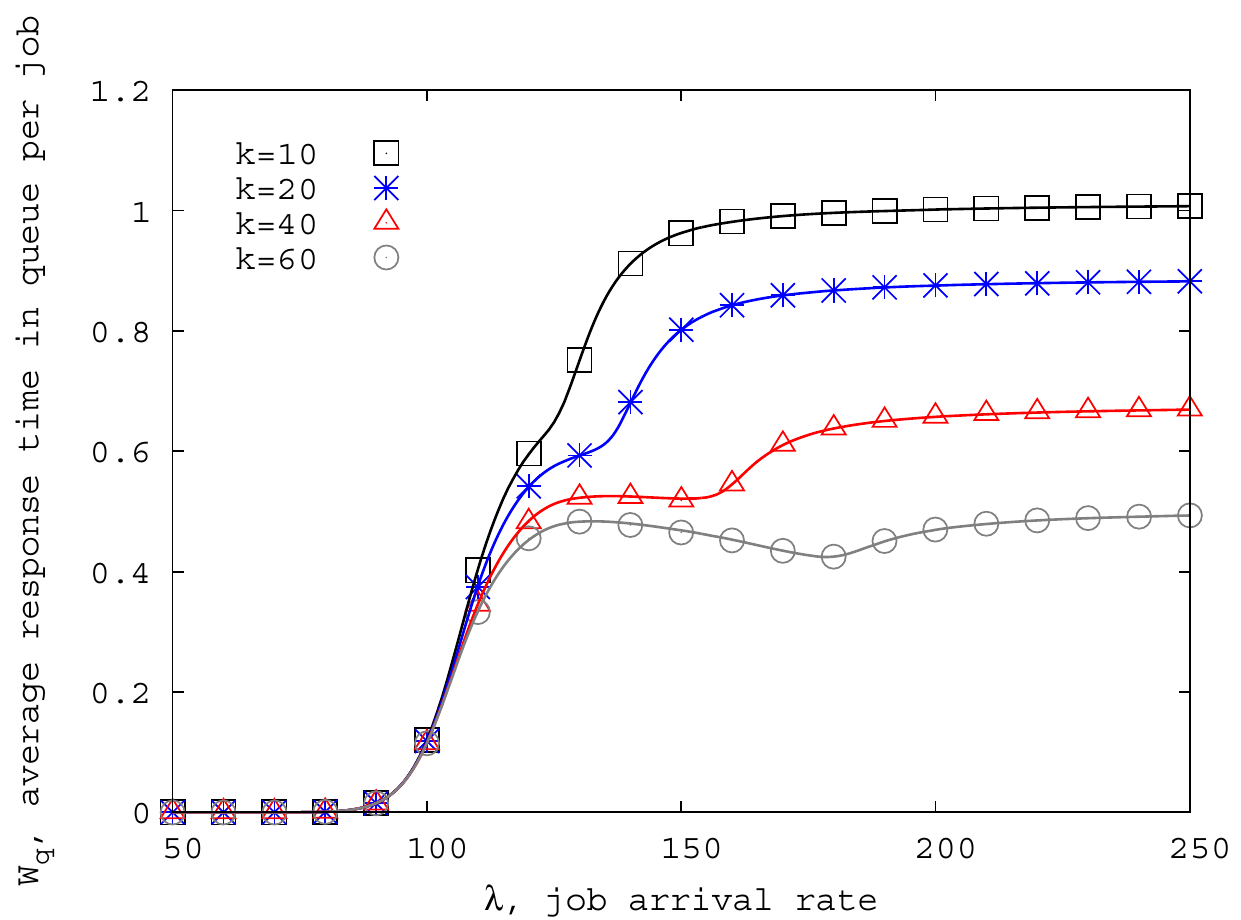}
			\caption{Impacts  on $Wq$.}
			\label{sfig:Wq_k}
		\end{subfigure}
			\begin{subfigure}[b]{0.45\textwidth}
				\includegraphics[width=8cm]{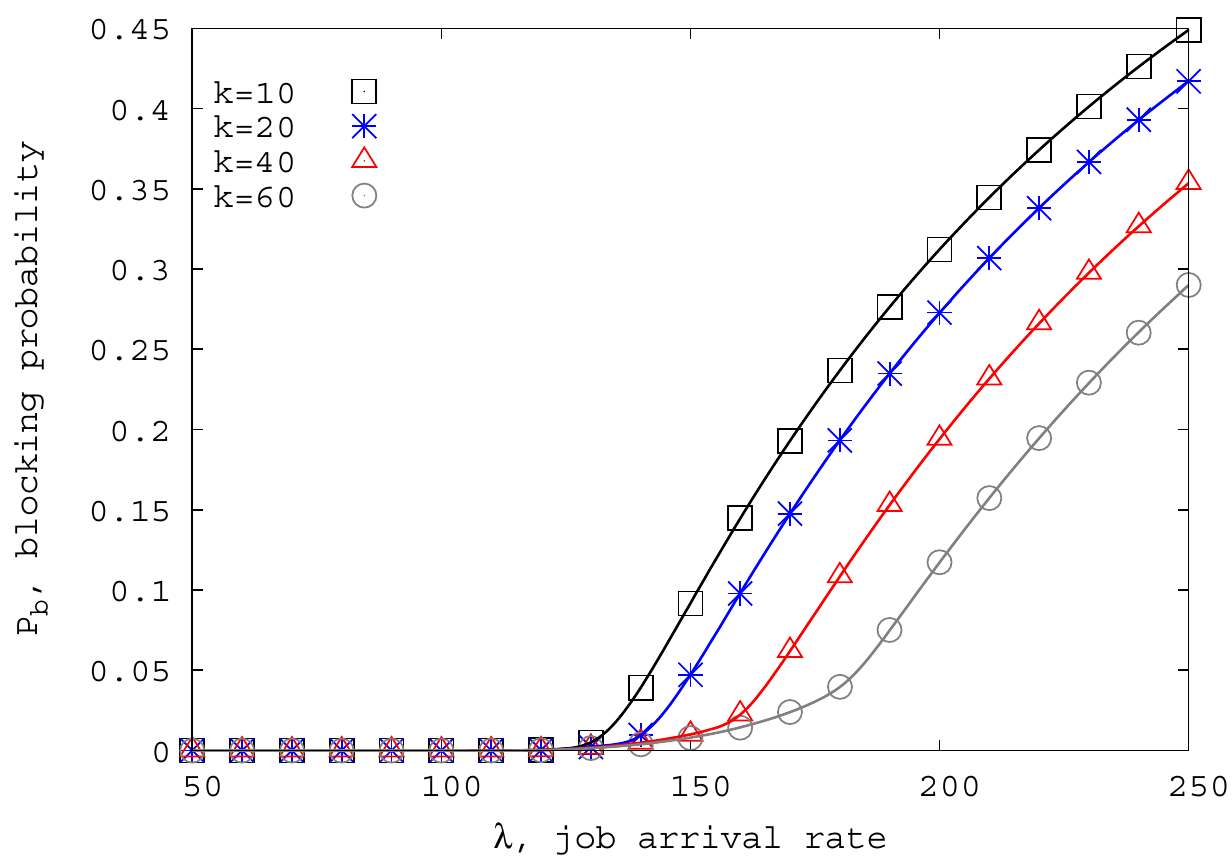}
				\caption{Impacts on $P_b$.}
				\label{sfig:Pb_k}
			\end{subfigure}
			\begin{subfigure}[b]{0.45\textwidth}
				\includegraphics[width=8cm]{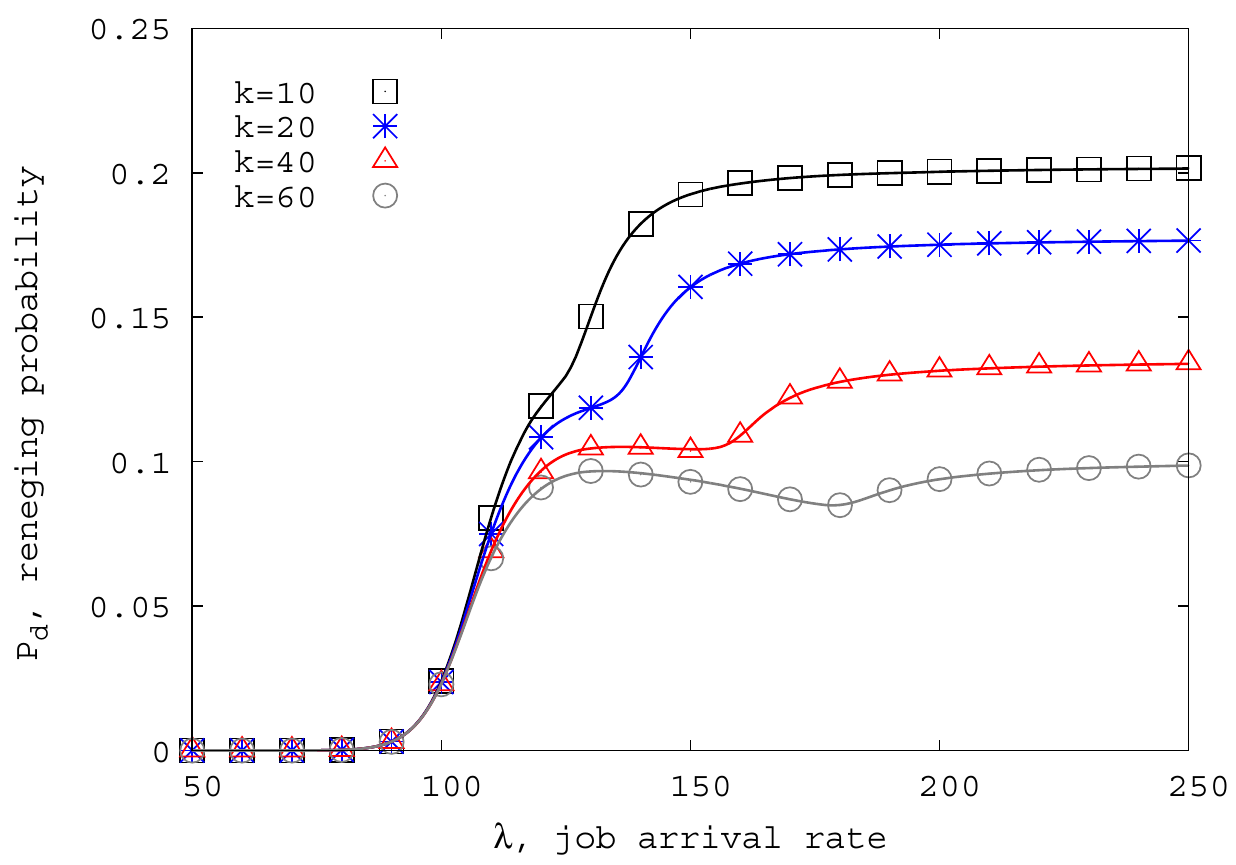}
				\caption{Impacts on $P_{d}$.}
				\label{sfig:Pd_k}
			\end{subfigure}
	\caption{Impacts of $k$ on the performance metrics ($n_0 = 100$).}
	\label{fig:Impacts_of_k}
\vspace{-4mm}
\end{figure*}
\begin{figure*}
	\centering
	\begin{subfigure}[b]{0.45\textwidth}
		\includegraphics[width=8cm]{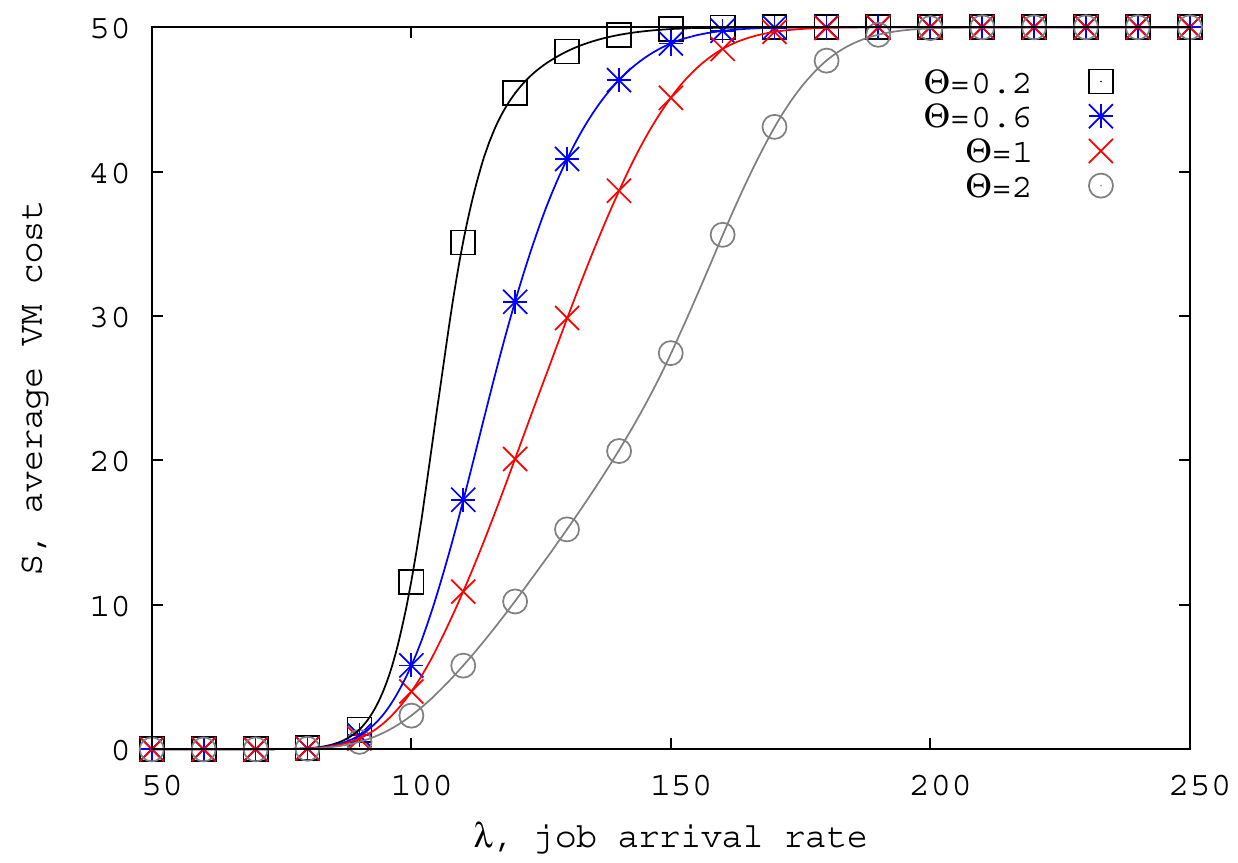}
		\caption{Impacts  on $S$.}
		\label{sfig:S_theta}
	\end{subfigure}
	\begin{subfigure}[b]{0.45\textwidth}
		\includegraphics[width=8cm]{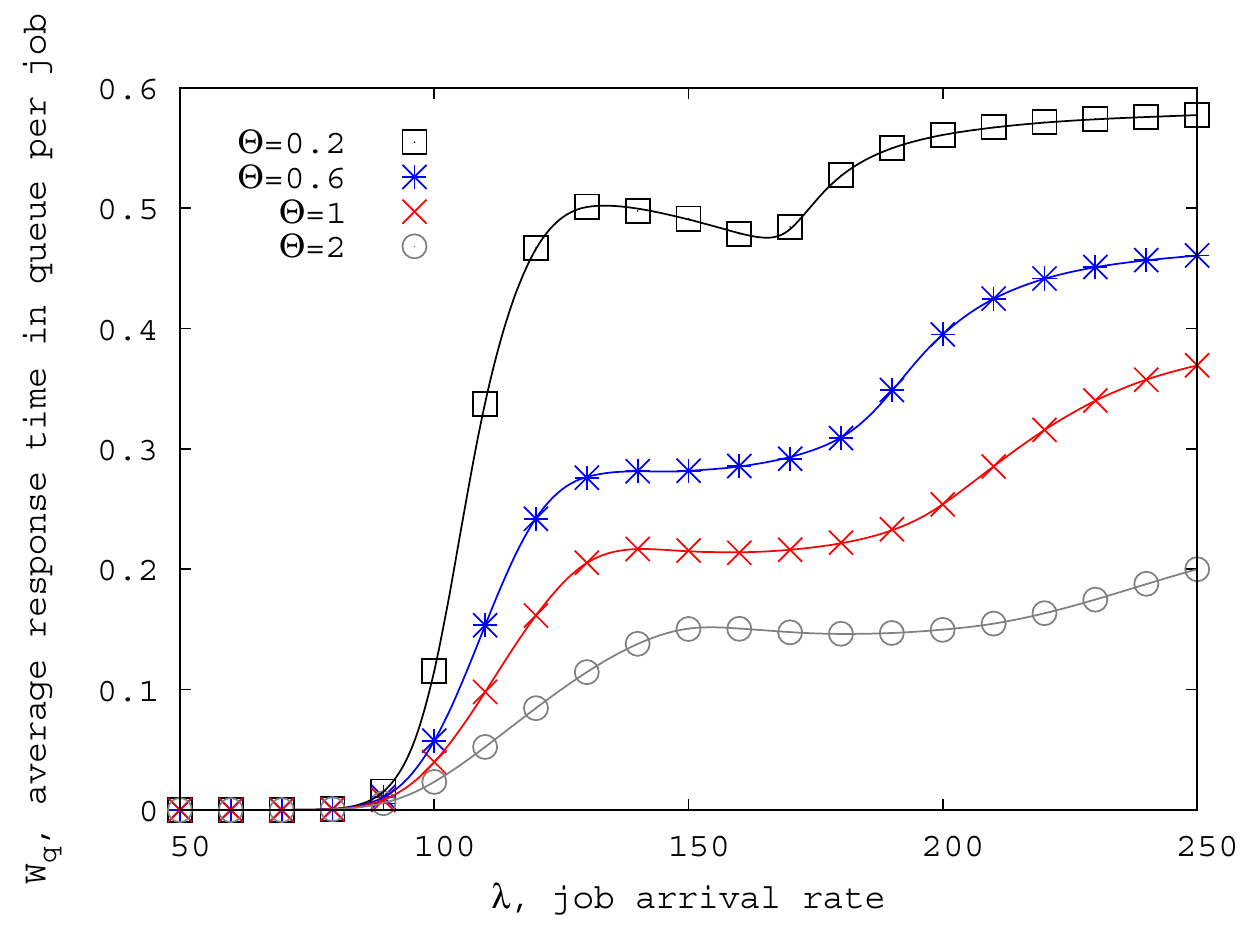}
		\caption{Impacts on $Wq$.}
		\label{sfig:Wq_theta}
	\end{subfigure}
	\begin{subfigure}[b]{0.45\textwidth}
		\includegraphics[width=8cm]{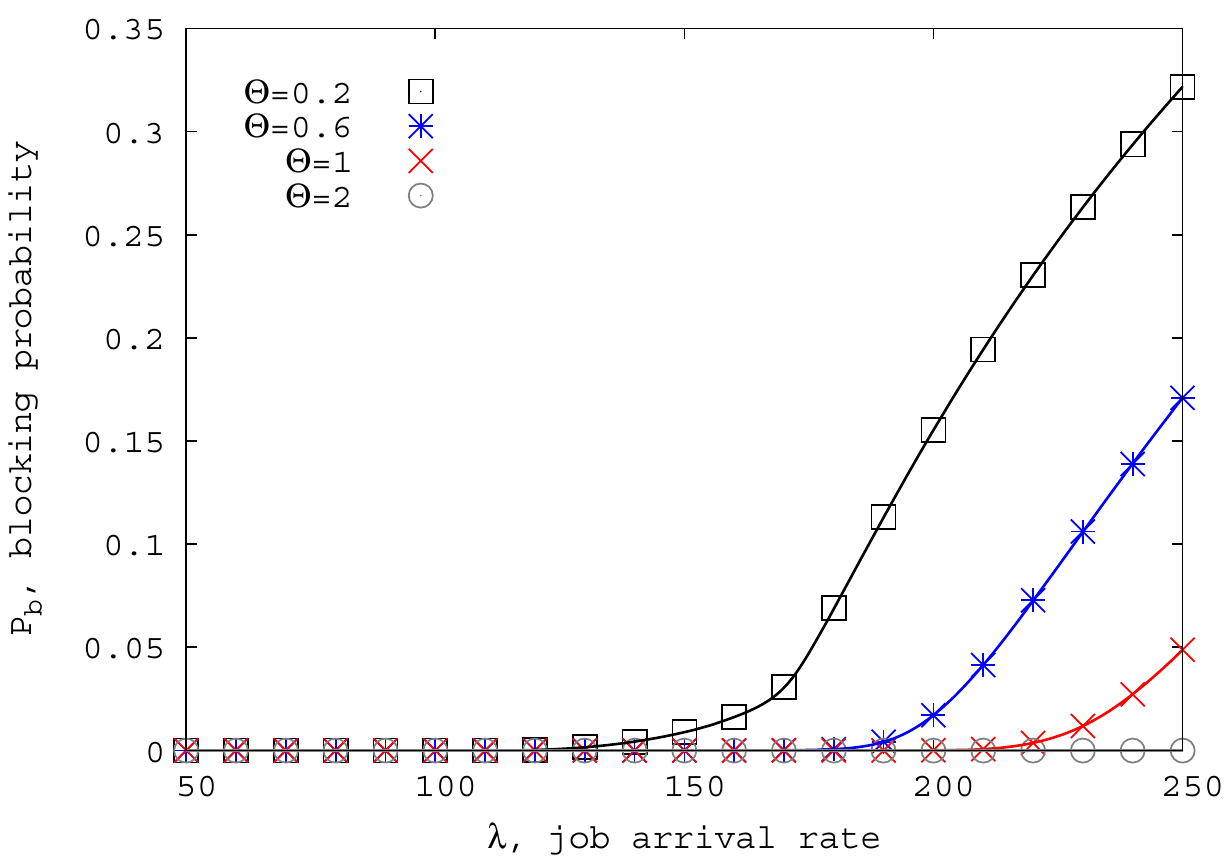}
		\caption{Impacts on $P_b$.}
		\label{sfig:Pb_theta}
	\end{subfigure}
	\begin{subfigure}[b]{0.45\textwidth}
		\includegraphics[width=8cm]{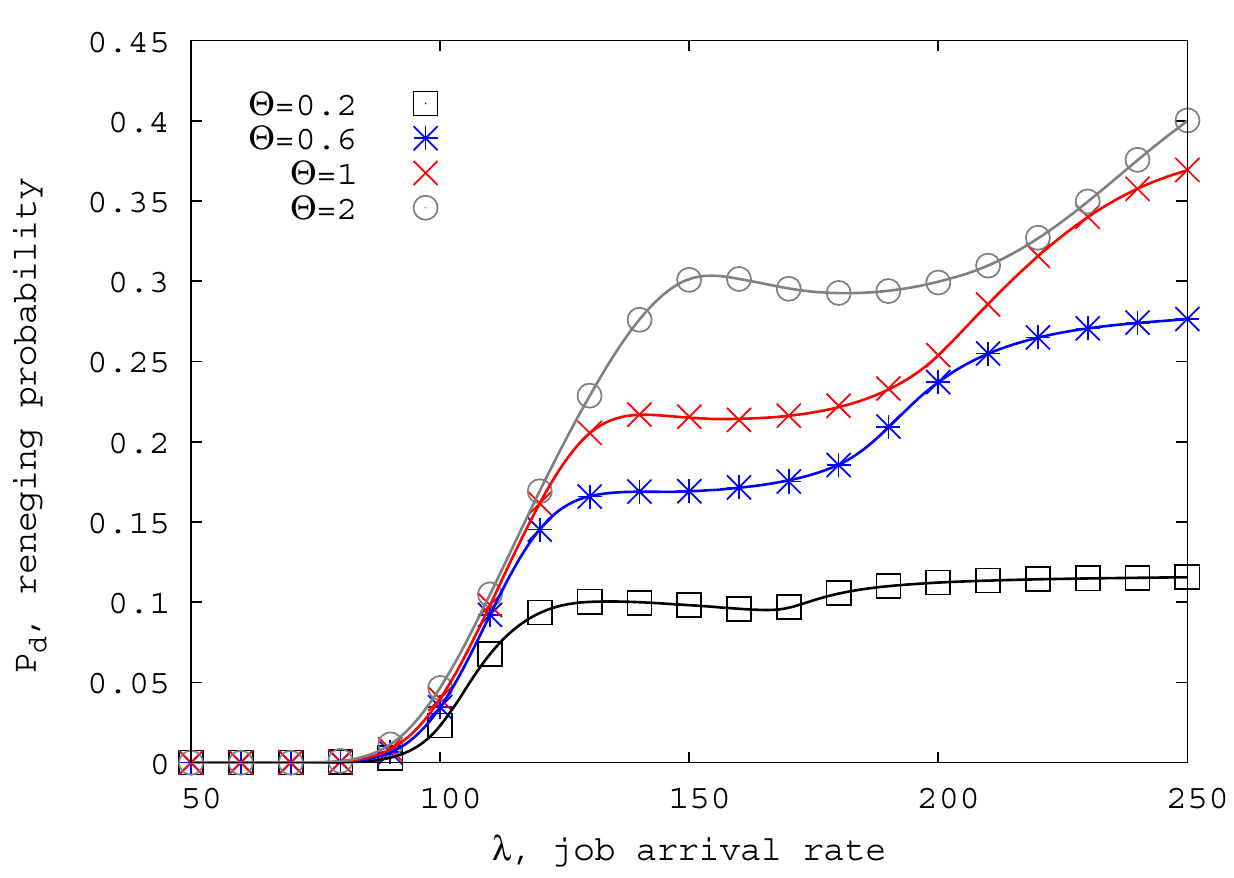}
		\caption{Impacts on $P_{d}$.}
		\label{sfig:Pd_theta}
	\end{subfigure}
	\caption{Impacts of $\theta$ on the performance metrics ($n_0=100$, $k = 50$).}
	\label{fig:Impacts_of_theta}
	\vspace{-4mm}
\end{figure*}
\begin{figure*}
	\centering
	\begin{subfigure}[b]{0.45\textwidth}
		\includegraphics[width=8cm]{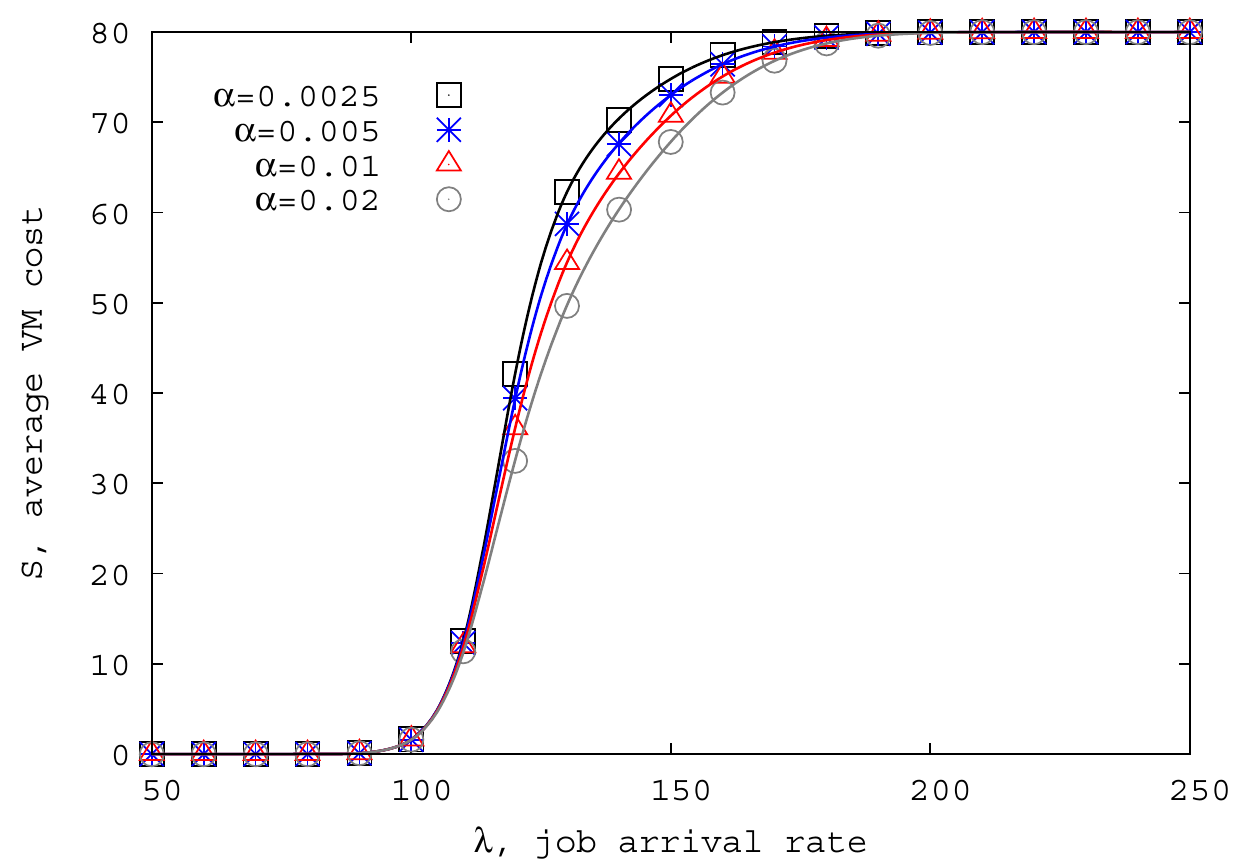}
		\caption{Impacts  on $S$.}
		\label{sfig:S_alpha}
	\end{subfigure}
	\begin{subfigure}[b]{0.45\textwidth}
		\includegraphics[width=8cm]{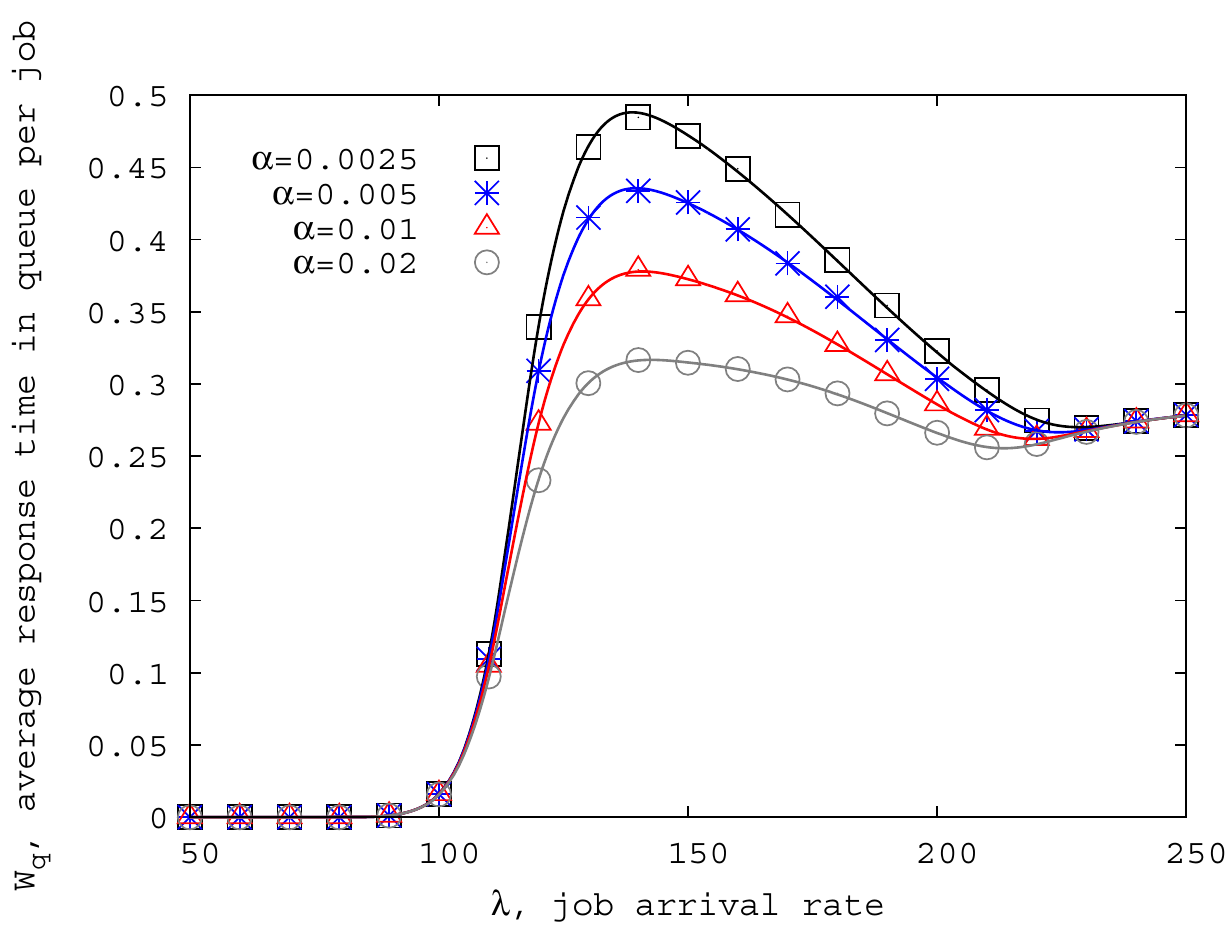}
		\caption{Impacts  on $Wq$.}
		\label{sfig:Wq_alpha}
	\end{subfigure}
		\begin{subfigure}[b]{0.45\textwidth}
			\includegraphics[width=8cm]{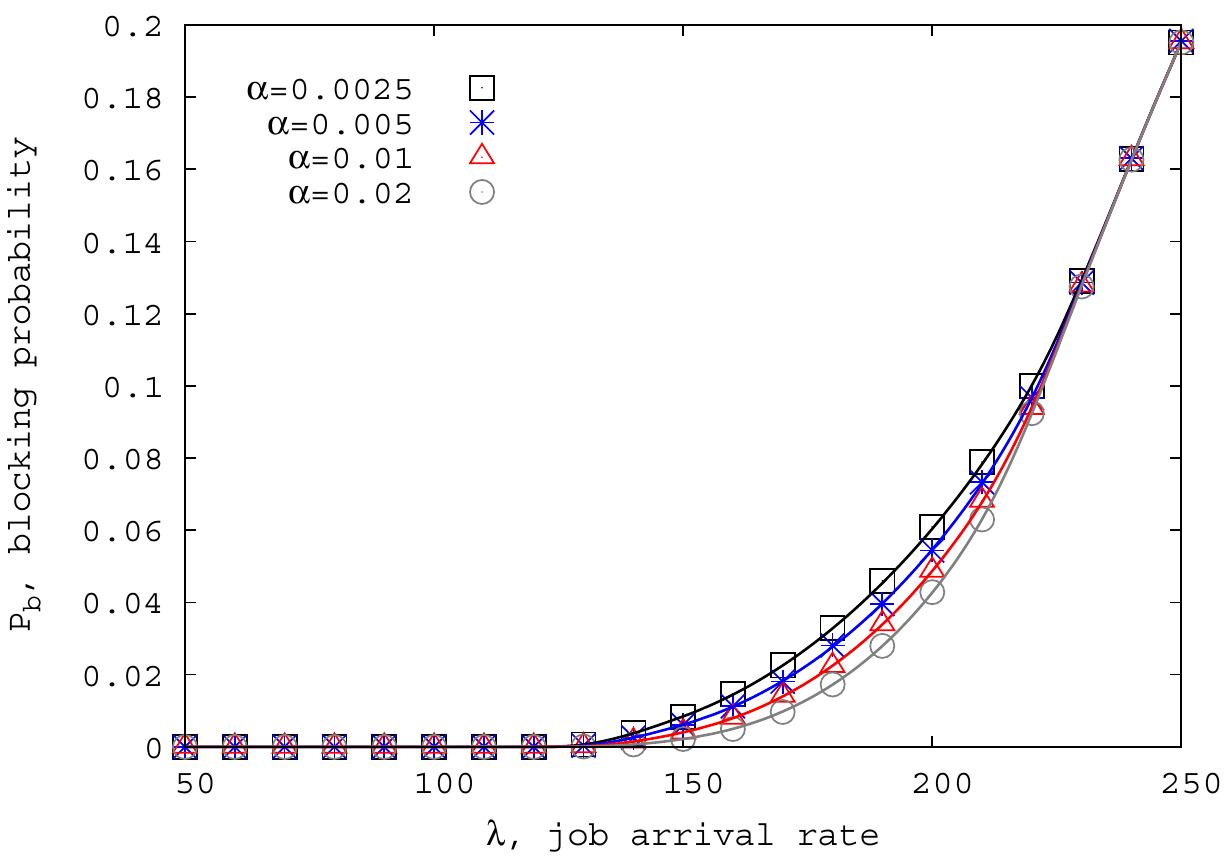}
			\caption{Impacts  on $P_b$.}
			\label{sfig:Pb_alpha}
		\end{subfigure}
		\begin{subfigure}[b]{0.45\textwidth}
			\includegraphics[width=8cm]{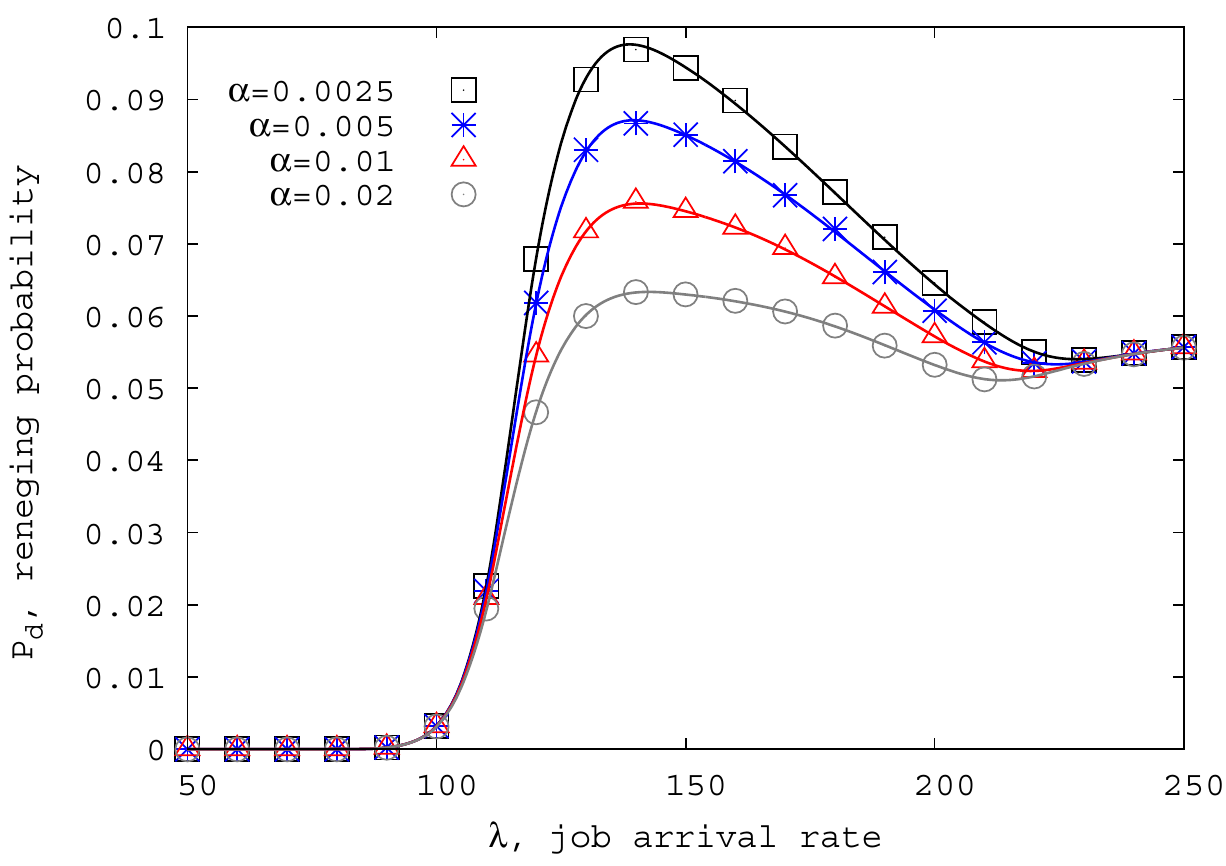}
			\caption{Impacts  on $P_{d}$.}
			\label{sfig:Pd_alpha}
		\end{subfigure}
	\caption{Impacts of $\alpha$ on the performance metrics  ($k = 80$).}
	\label{fig:Impacts_of_alpha}
\vspace{-4mm}
\end{figure*}
\begin{figure*}
	\centering
	\begin{subfigure}[b]{0.45\textwidth}
		\includegraphics[width=8cm]{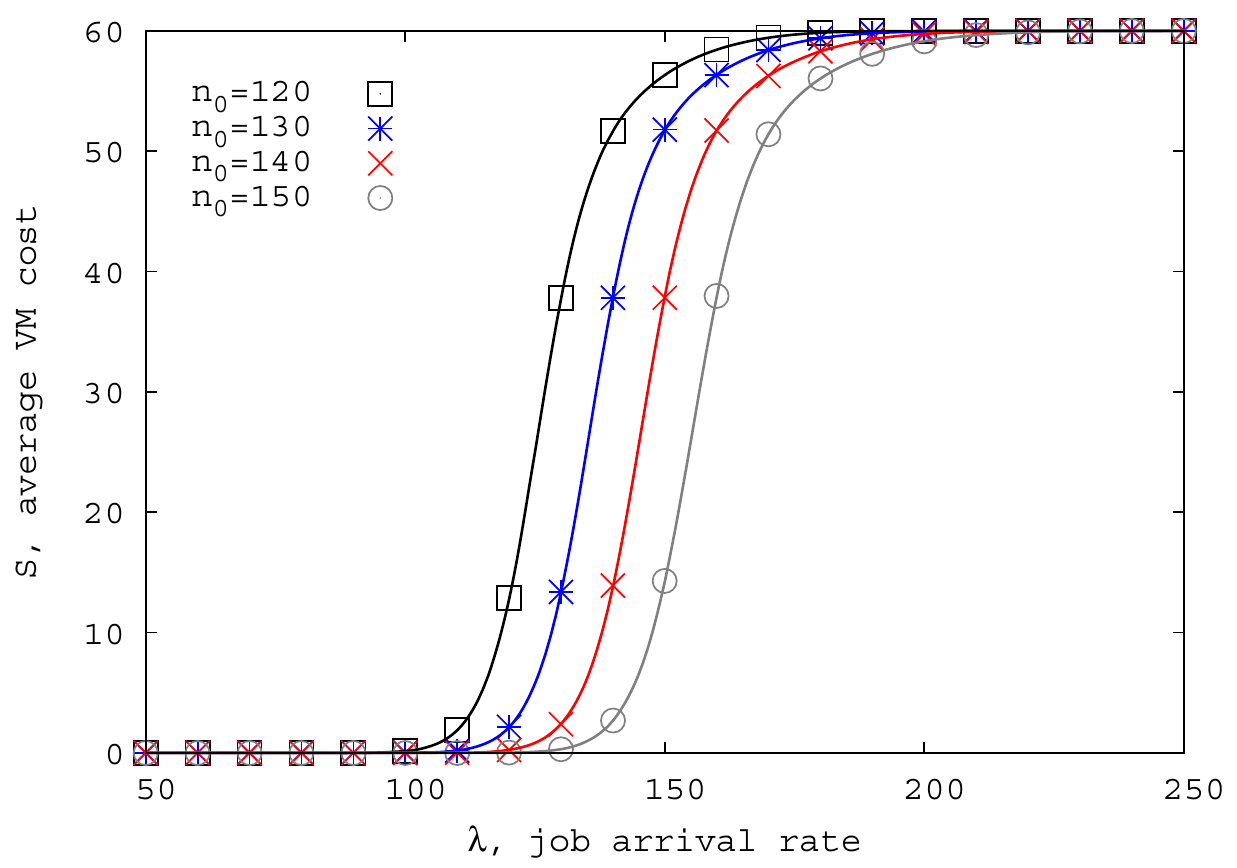}
		\caption{Impacts  on $S$.}
		\label{sfig:S_n0}
	\end{subfigure}
		\begin{subfigure}[b]{0.45\textwidth}
			\includegraphics[width=8cm]{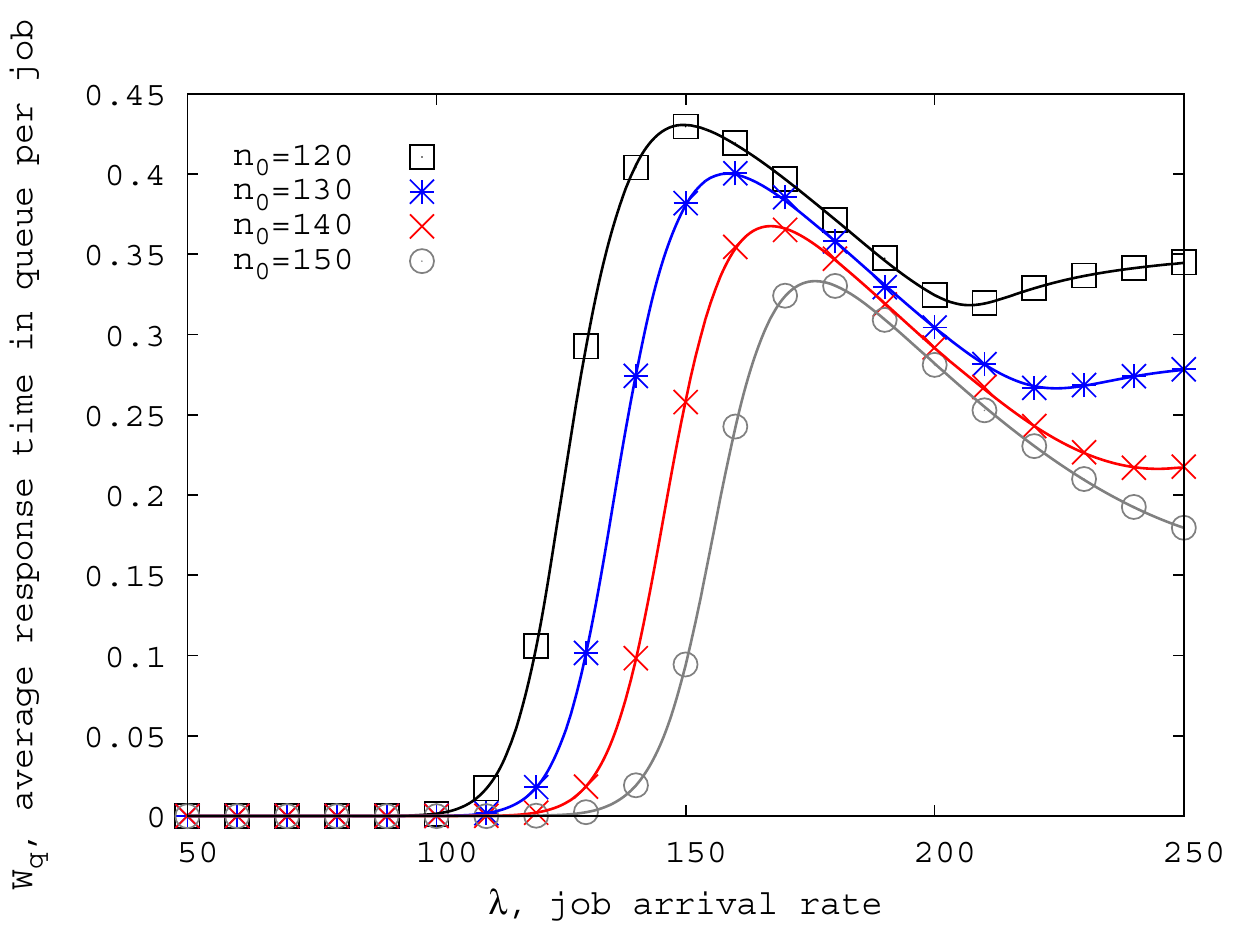}
			\caption{Impacts  on $Wq$.}
			\label{sfig:Wq_n0}
		\end{subfigure}
	\begin{subfigure}[b]{0.45\textwidth}
		\includegraphics[width=8cm]{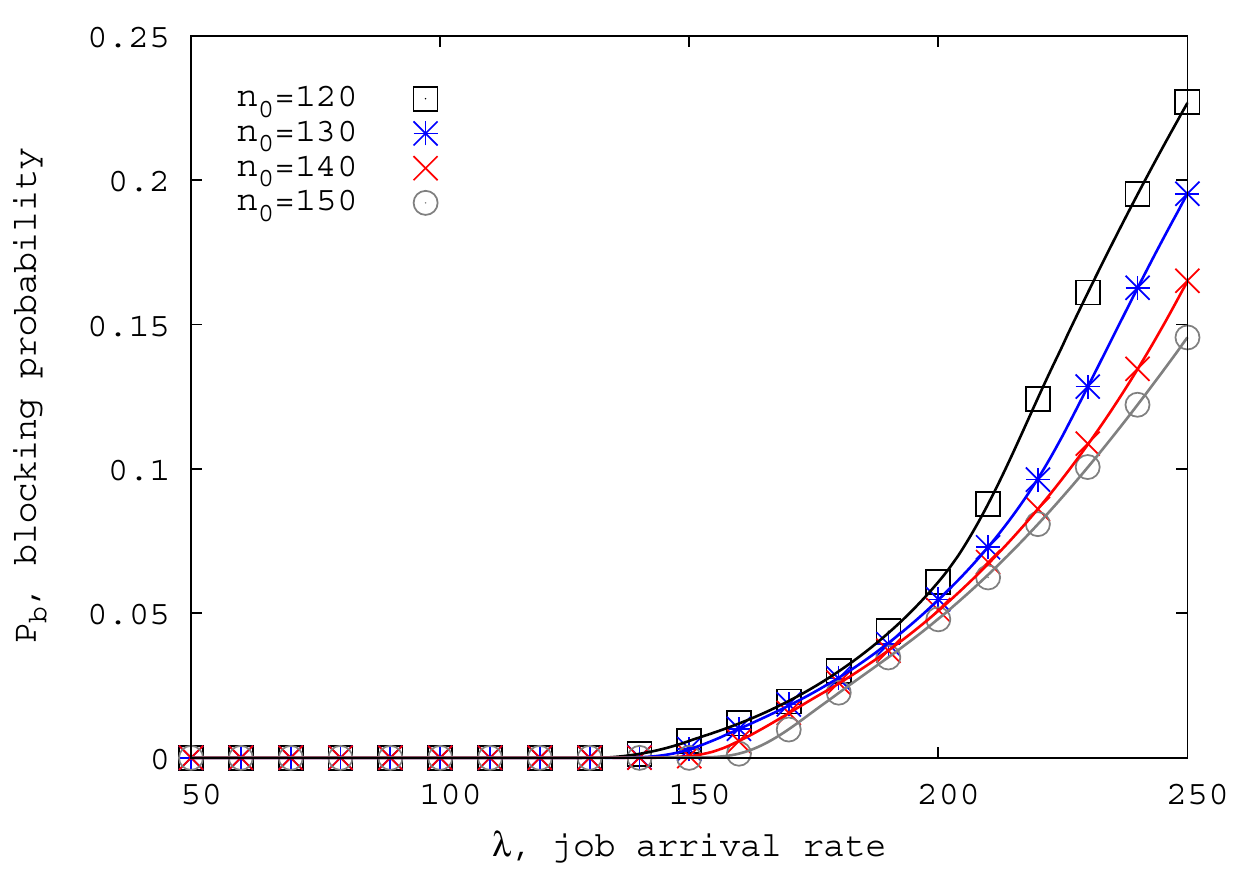}
		\caption{Impacts  on $P_b$.}
		\label{sfig:Pb_n0}
	\end{subfigure}
	\begin{subfigure}[b]{0.45\textwidth}
		\includegraphics[width=8cm]{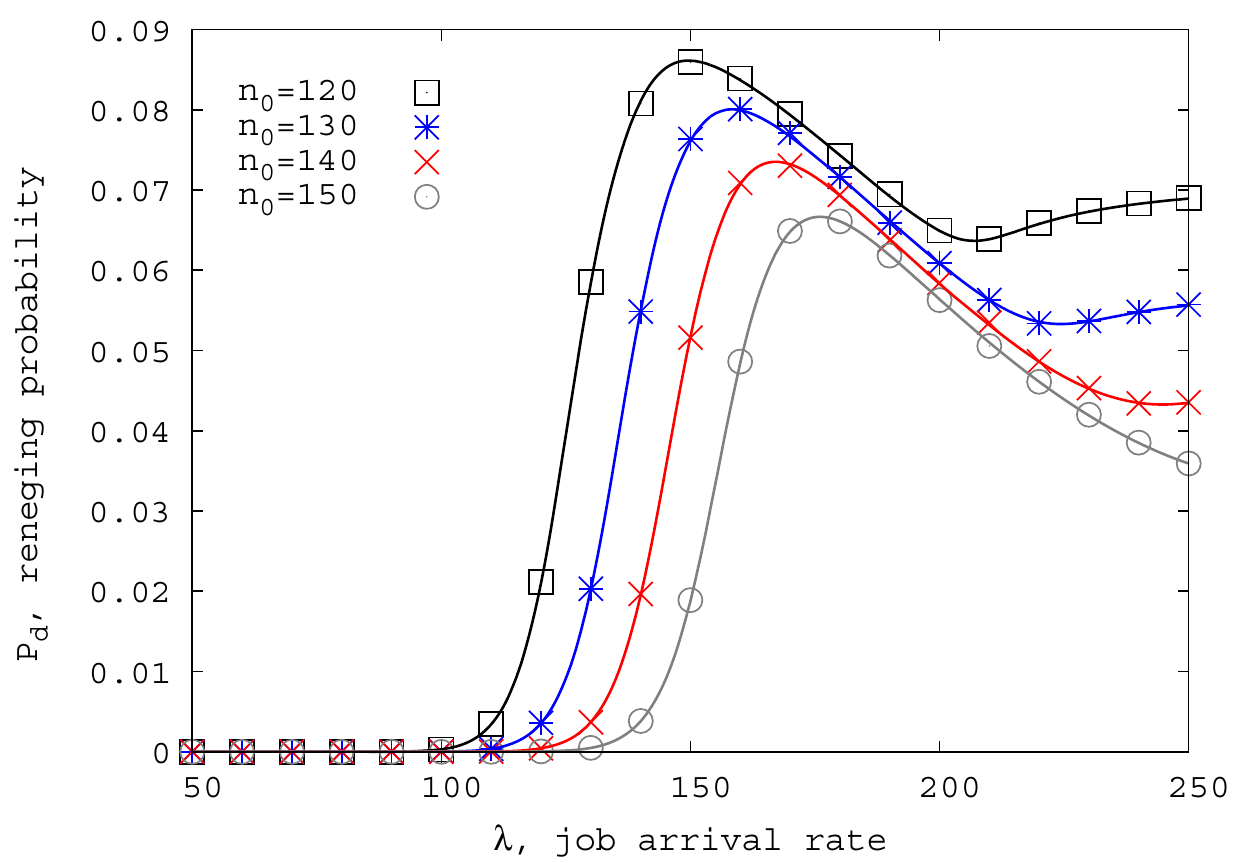}
		\caption{Impacts  on $P_{d}$.}
		\label{sfig:Pd_n0}
	\end{subfigure}
	\caption{Impacts of $n_0$ on the performance metrics ($k = 60$).}
	\label{fig:Impacts_of_n0}
\vspace{-4mm}
\end{figure*}
\begin{figure*}
	\centering
	\begin{subfigure}[b]{0.45\textwidth}
		\includegraphics[width=8cm]{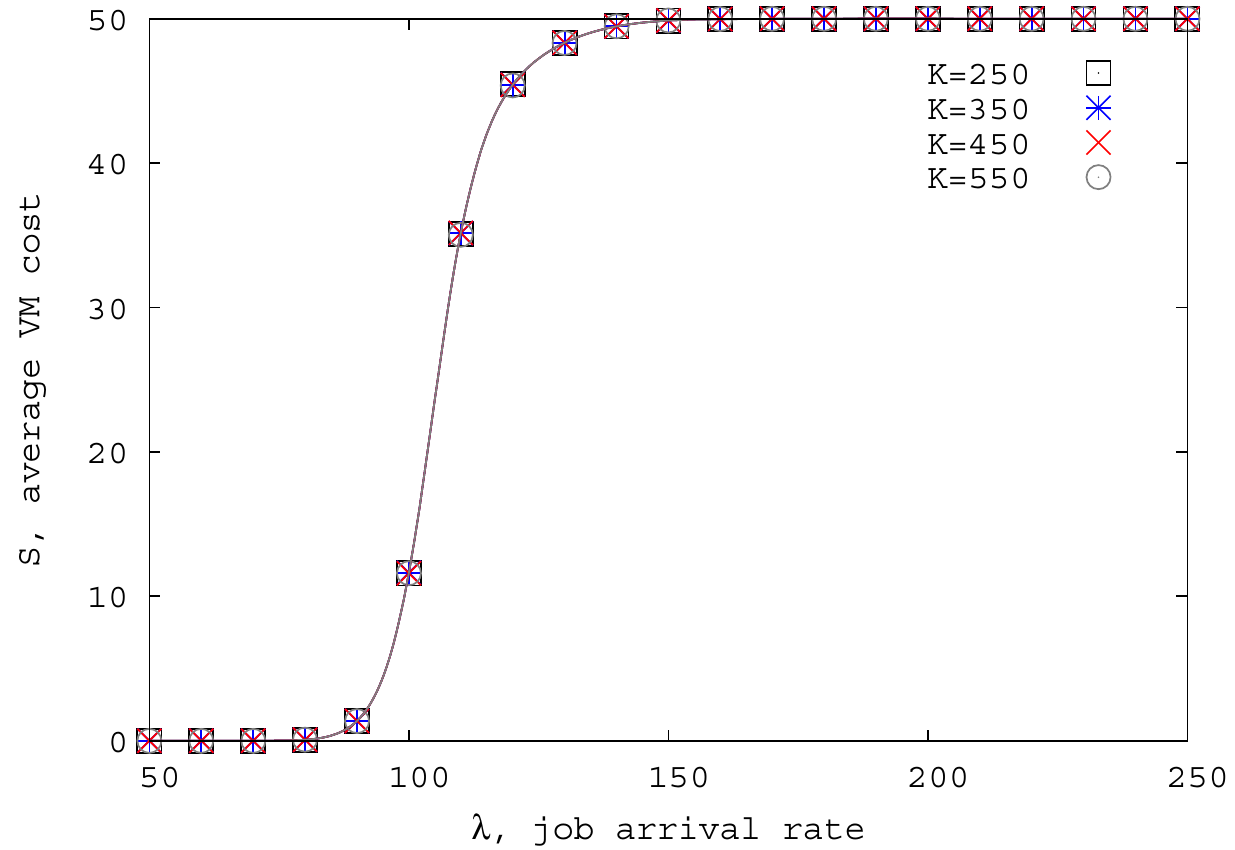}
		\caption{Impacts  on $S$.}
		\label{sfig:S_K}
	\end{subfigure}
	\begin{subfigure}[b]{0.45\textwidth}
		\includegraphics[width=8cm]{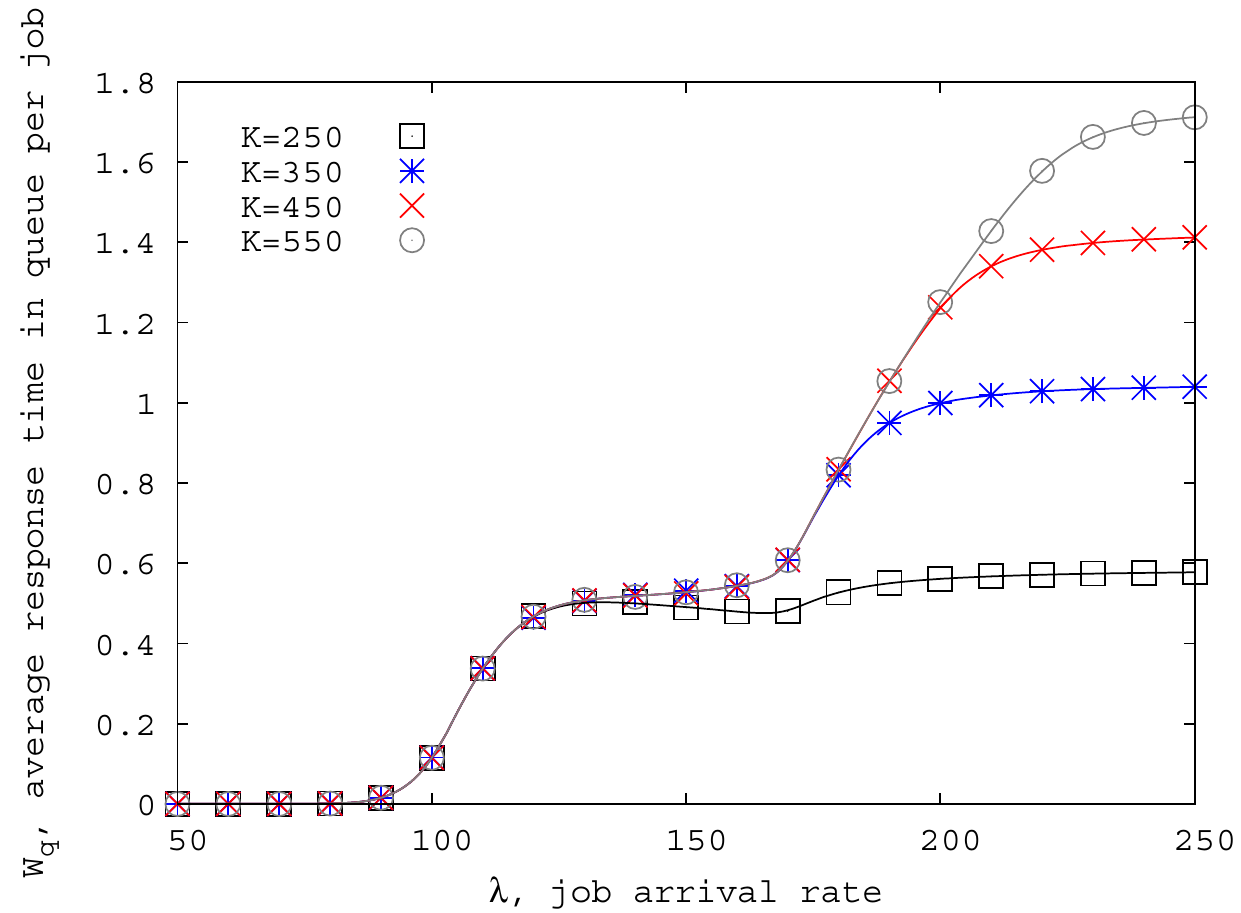}
		\caption{Impacts  on $Wq$.}
		\label{sfig:Wq_K}
	\end{subfigure}
	\begin{subfigure}[b]{0.45\textwidth}
		\includegraphics[width=8cm]{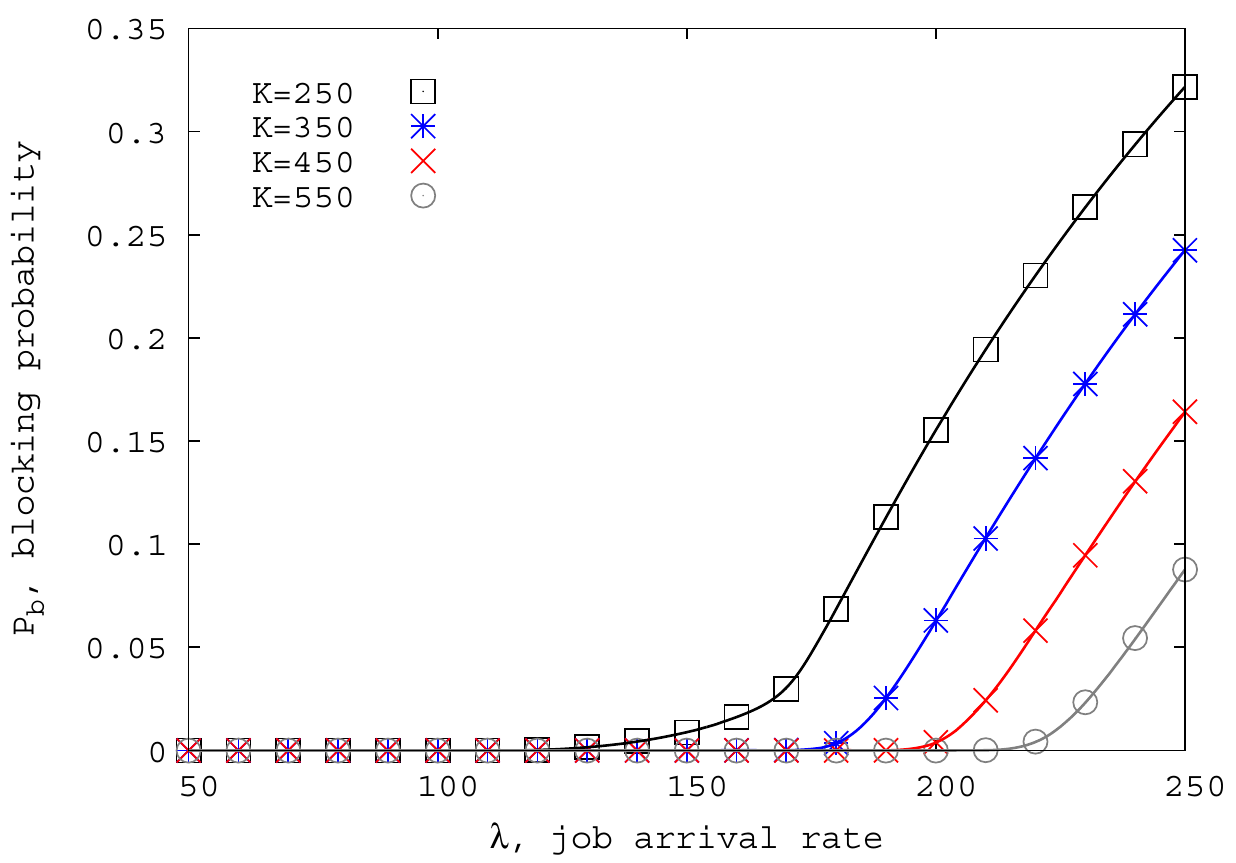}
		\caption{Impacts  on $P_b$.}
		\label{sfig:Pb_K}
	\end{subfigure}
	\begin{subfigure}[b]{0.45\textwidth}
		\includegraphics[width=8cm]{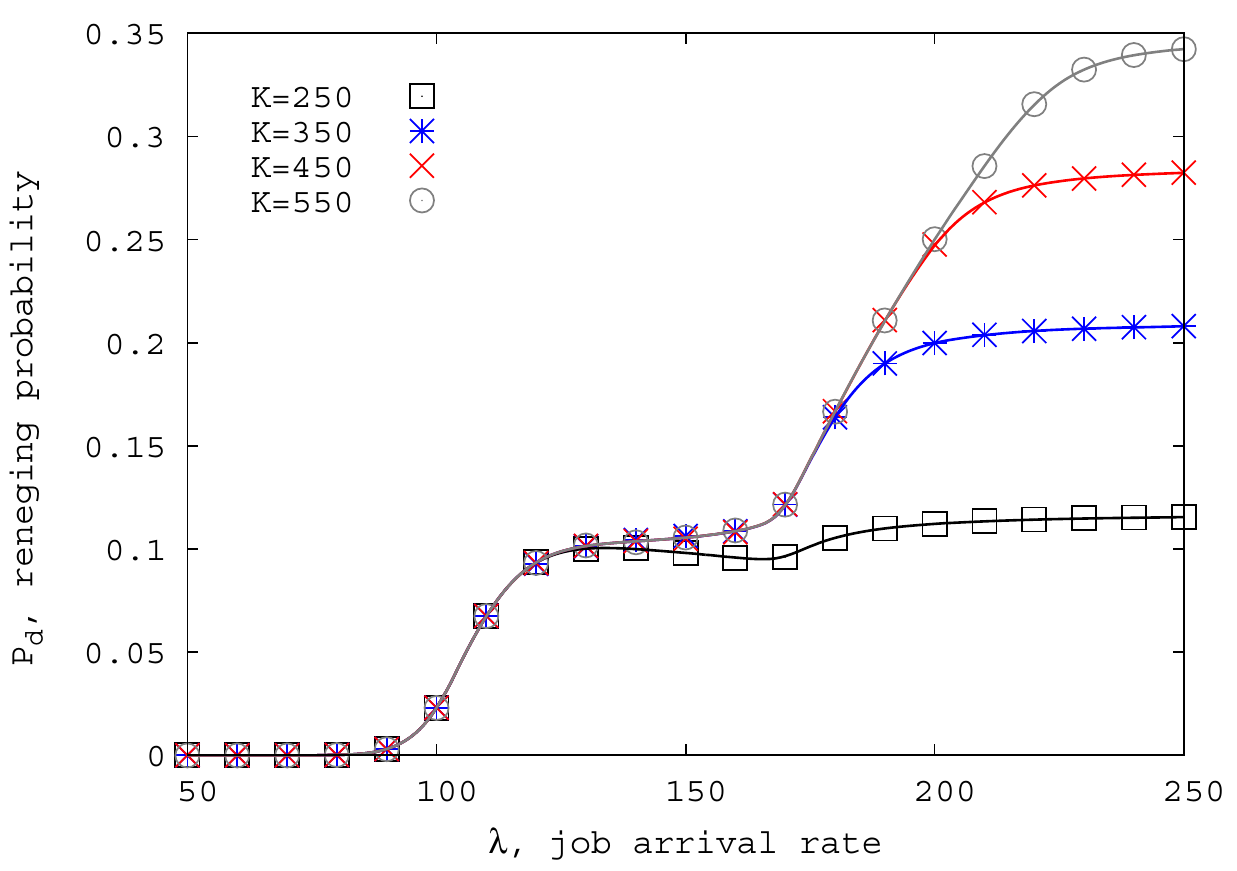}
		\caption{Impacts  on $P_{d}$.}
		\label{sfig:Pd_K}
	\end{subfigure}
	\caption{Impacts of $K$ on the performance metrics  ($k = 50$).}
	\label{fig:Impacts_of_K}
	\vspace{-4mm}
\end{figure*}

Figs.~\ref{fig:Impacts_of_k}, \ref{fig:Impacts_of_theta}, \ref{fig:Impacts_of_alpha}, \ref{fig:Impacts_of_n0},  \ref{fig:Impacts_of_K} not only demonstrate the correctness of our analytical model, but also illustrate the impacts of $\lambda$, $k$, $\theta$,  $\alpha$, $n_0$, $K$ on the performance metrics: average VM cost $S$, average response time in queue $W_q$, blocking probability $P_b$, and dropping probability $P_d$, respectively. \emph{In the figures, the \textit{lines} denote analytical results, and the \textit{points} represent simulation results.}  Each simulation result in the figures is the mean value of the results in 300,000 seconds with 95\% confidence level.

\subsection{Impacts of Arrival Rate $\lambda$}
We first look into the impacts of job request arrival rate $\lambda$. Mobile operators cannot adjust $\lambda$ but are able to monitor it and configure network parameters $k$, $\theta$,   $\alpha$, $n_0$, and $K$ for network optimization accordingly.

Figs.~\ref{sfig:S_k}, \ref{sfig:S_theta}, \ref{sfig:S_alpha}, \ref{sfig:S_n0}, \ref{sfig:S_K} depict the impacts of $\lambda$ on  $S$. In general, one can see that $S$ initiates at 0 at the beginning and then starts to raise sharply when $\lambda$ passes $n_o\mu$.  The reason is that the incoming job requests are served by the legacy equipment when $\lambda < n_0\mu$. No VMs are powered on. Then DBCA starts to turn on VMs to handle job requests as $\lambda$ is increasing.  Later, $S$ reaches at a bound even if $\lambda$ continues growing. This is because all the $k$ VMs are turned on so that $S$ is bounded as $k$ VM costs.

Figs.~\ref{sfig:Wq_k}, \ref{sfig:Wq_theta}, \ref{sfig:Wq_alpha}, \ref{sfig:Wq_n0}, \ref{sfig:Wq_K} show the impacts of $\lambda$ on $W_q$. Interestingly, the trend of the curves can generally be  divided into four phases: zero phase, ascent phase, descent phase, and saturation phase. In the zero phase, $W_q$ is zero because incoming jobs are served immediately by available capacities. In the ascent phase, $W_q$ raises sharply due to the setup time of VMs. Specifically, when $\lambda$ approaches to $n_0\mu$ and then larger than $n_0\mu$, VMs start to be powered on and to serve jobs. In doing so, however, $W_q$ still grows sharply because jobs have to wait for turning on processes of VMs. Later, $W_q$ starts to decrease due to new running VMs as shown as the third (descent) phase.  In the forth (ascent) phase, $W_q$ starts to grow again and then saturates at a bound. The reason of ascent is that the system is not able to serve the coming jobs when $\lambda \geq (n_0+k)\mu$. Finally, the curves reach to saturation because the capacity of the system is too full to handle the jobs and the value of $W_q$ is limited by the total system capacity $K$.

In Figs.~\ref{sfig:Pb_k}, \ref{sfig:Pb_theta}, \ref{sfig:Pb_alpha}, \ref{sfig:Pb_n0}, and~\ref{sfig:Pb_K}, we study the impacts of $\lambda$ on $P_b$. The trends of the curves are relatively simple compared with the above two metrics. Generally, the curves are growing as $\lambda$ increases. In particular, $P_b$ initiates at 0 and starts to increase when $\lambda>(n_0+k)\mu$. The reason is that the system starts to reject jobs when the queue is full.

Figs.~\ref{sfig:Pd_k}, \ref{sfig:Pd_theta}, \ref{sfig:Pd_alpha}, \ref{sfig:Pd_n0}, \ref{sfig:Pd_K} illustrate the impacts of $\lambda$ on $P_d$. One can see that the trends of the curves are similar with that of $W_q$. Note that job requests start to quit the queue if the waiting time exceeds their deadline constraints. So $P_d$ is highly related to $W_q$. If $W_q$ is large then jobs are dropped with high probability. This also explains why the trends are similar. Please refer to the above discussion of  $W_q$ for  $P_d$.

\subsection{Impacts of the Number of VNF Instances $k$} \label{ssec:Impacts_of_k}
The figures in Fig.~\ref{fig:Impacts_of_k} depict the impacts of $k$ on performance metrics $S$, $W_q$, $P_b$, and $P_d$, respectively. We can see that increasing $k$ from 10 to 60 leads to the gains of $S$ while decreasing $W_q$, $P_b$, and $P_d$ accordingly. A larger $k$ means that more VMs could be used to handle the growing job requests so $W_q$, $P_b$, and $P_d$ are improved. If a operator wants to adjust budget constraint $S$, the operator can specify a suitable $k$ based on (\ref{eq:S}).

\subsection{Impacts of Abandon Rate $\theta$}
In Figs.~\ref{sfig:S_theta},  \ref{sfig:Wq_theta}, \ref{sfig:Pb_theta}, and~\ref{sfig:Pd_theta}, we study the impacts of abandon rate $\theta$ on $S$, $W_q$, $P_b$, and $P_d$, respectively. Recall that a job request is assumed to have a deadline constraint with mean $1/\theta$, meaning that the job will stop waiting in the queue if the waiting time exceeds its deadline. We observe that increasing $\theta$ decreases $S$, $W_q$, and $P_b$ while enlarging $P_d$. Specifically, as shown in Fig.~\ref{sfig:S_theta}, $\theta$ has no impacts on $S$ when $\lambda < n_0\mu$ or $\lambda>(n_0+k)\mu$. The reason is that $S$ only depends on the number of running VMs. Whereas, when $n_0\mu < \lambda<(n_0+k)\mu$, a larger $\theta$ leads to less $S$ because more jobs are dropped from the system. In addition, the impacts of $\theta$ on $W_q$ is illustrated in Fig.~\ref{sfig:Wq_theta}. A larger $\theta$ makes a smaller $W_q$. The reason is that when more jobs quit from the queue, the rest of the jobs need to wait less time. Fig.~\ref{sfig:Pb_theta} shows that increasing $\theta$ leads to less $P_b$. The reason is straightforward. More jobs quitting from the queue means that the system has more available capacities to handle the incoming jobs.  In Fig.~\ref{sfig:Pd_theta}, we observe that a larger $\theta$ means more $P_d$. It coincides with the definition of $P_d$.

\subsection{Impacts of VM Setup Rate $\alpha$}
Figs.~\ref{sfig:S_alpha},  \ref{sfig:Wq_alpha}, \ref{sfig:Pb_alpha}, \ref{sfig:Pd_alpha} illustrate the impacts of $\alpha$ on $S$, $W_q$, $P_b$, and $P_d$, respectively. Recall that VMs are assumed to have a setup time with mean value $1/\alpha$. To reduce the setup time, NFV Management and Orchestration can perform scale-up procedure to add resources (e.g., CPU, memory) to make VMs more powerful. We observe that less setup time decreases $S$, $W_q$, $P_b$, and $P_d$. The reason is that short setup time leads to that VMs can be quicker to be available for handling the jobs, resulting in less operation cost (see Fig.~\ref{sfig:S_alpha}), lower waiting time for jobs (see Fig.~\ref{sfig:Wq_alpha}), smaller blocking probability (see Fig.~\ref{sfig:Pd_alpha}), and reduced dropping probability as shown in Fig.~\ref{sfig:Pd_alpha}.

\subsection{Impacts of Capacities of Legacy Equipment $n_0$}
Figs.~\ref{sfig:S_n0},  \ref{sfig:Wq_n0}, \ref{sfig:Pb_n0}, \ref{sfig:Pd_n0} show the impacts of $n_0$ on $S$, $W_q$, $P_b$, and $P_d$, respectively. We observe that the curves initiate at 0 then stay at 0 for a period and start to grow up as $\lambda$ increases. $n_0$ decides the length of the period when the curves start to ascend. The reason is that the legacy equipment can handle incoming jobs within its capacity. When $\lambda$ exceeds the capacity of the legacy equipment, the performance metrics $S$, $W_q$, $P_b$, and $P_d$ start to grow up.

\subsection{Impacts of System Capacity $K$}
In Figs.~\ref{sfig:S_K},  \ref{sfig:Wq_K}, \ref{sfig:Pb_K}, and~\ref{sfig:Pd_K}, we investigate the impacts of $K$ on $S$, $W_q$, $P_b$, and $P_d$, respectively. As shown in Fig.~\ref{sfig:S_K}, we observe that $K$ has limited impacts on $S$. As we discussed in Section~\ref{ssec:Impacts_of_k}, $S$ is mainly decided by $k$. Figs.~\ref{sfig:Wq_K}, \ref{sfig:Pb_K}, and~\ref{sfig:Pd_K} show that  $K$ has significant impacts on $W_q$, $P_b$ as well as $P_d$. Different $K$ makes huge gaps between the curves.  Moreover, a large $K$ leads to a larger $W_q$ as well as $P_d$ but makes $P_b$ smaller. The reason is that it enables more jobs waiting in the queue rather than dropping them.

\section{Conclusions} \label{sec:Conclusions}
In this paper, we have proposed DBCA for addressing the tradeoff between operation budget constraint $S$ and system performance which is evaluated by three performance metrics: the average job response time $W_q$,  blocking probability $P_b$, and dropping probability $P_d$.  Our work addresses the research gap by considering both VM setup time and the capacity of legacy equipment  in NFV enabled EPC scenarios.  Compared with our previous work~\cite{YiRen2016globecom}, the model quantifies a more practical case.  Our results show that the analytical model provides a quick way to help mobile operators to plan and design network optimization strategies without wide deployment, saving on cost and time. Moreover, based on our analytical model, mobile operators can easily estimate operation budget given desired system performance, vice versa.

As our future work, one extension is to generalize the VM setup time and the arrival and service time. Right now there is no literature to support that they are exponential random variables. These results could be generalized by using orthogonal polynomial approaches~\cite{pender2014gram}. Also, we plan to relax the assumption of VM scaling in/out capability, i.e., from one VNF instance per time to arbitrary instances per time.  We plan to complete these works in  follow-up papers.

\bibliographystyle{IEEEtran}
\bibliography{IEEEACM,NFV}
\balance

\end{document}